\documentclass[10pt,twocolumn,journal]{IEEEtran}
\usepackage{cite}
\usepackage{pslatex}
\usepackage{moreverb}
\usepackage{epsfig}
\usepackage{amsmath}
\usepackage{booktabs}
\usepackage{array}
\usepackage{multirow}
\usepackage{threeparttable}
\usepackage{cite}
\usepackage{amsfonts}
\usepackage{dsfont,color}
\usepackage{amssymb}

\newtheorem{Theo}{Theorem}
\newtheorem{lemma}{Lemma}

\newtheorem{coro}{Corollary}
\newtheorem{proposition}{Proposition}
\newtheorem{Def}{Definition}

\DeclareMathAlphabet{\mathpzc}{OT1}{pzc}{m}{it}

\def\calW{{\mathcal W}}

\def\bg{{\mathbf g}}
\def\bh{{\mathbf h}}

\def\bn{{\mathbf n}}

\def\br{{\mathbf r}}

\def\bx{{\mathbf x}}
\def\by{{\mathbf y}}
\def\bz{{\mathbf z}}
\def\bA{{\mathbf A}}
\def\bB{{\mathbf B}}

\def\bD{{\mathbf D}}

\def\bG{{\mathbf G}}
\def\bH{{\mathbf H}}
\def\bI{{\mathbf I}}

\def\bS{{\mathbf S}}

\def\bU{{\mathbf U}}

\def\diag{{\rm diag}}
\def\tr{{\rm tr}}

\def\E{{\mathds E}}

\def\b0{{\mathbf 0}}

\begin{document}

\title{On Secrecy Capacity of Fast Fading MIMOME Wiretap Channels With Statistical CSIT}

\author{
\authorblockN{Shih-Chun Lin and Cheng-Liang Lin}
\thanks{S.-C. Lin and C.-L. Lin are with
the Department of Electronic and Computer Engineering, National
Taiwan University of Science and Technology, Taipei, Taiwan,
10607. E-mail:sclin@mail.ntust.edu.tw. This work was supported by
the National Science Council, Taiwan, R.O.C., under grant NSC
101-2221-E-027-085-MY3}  }

\IEEEoverridecommandlockouts
\maketitle
\begin{abstract}
In this paper, we consider secure transmissions in ergodic
Rayleigh fast-faded multiple-input multiple-output
multiple-antenna-eavesdropper (MIMOME) wiretap channels with only
statistical channel state information at the transmitter (CSIT).
When the legitimate receiver has more (or equal) antennas than the
eavesdropper, we prove the first MIMOME secrecy capacity with
partial CSIT by establishing a new secrecy capacity upper-bound.
The key step is to form an MIMOME degraded channel by dividing the
legitimate receiver's channel matrix into two submatrices, and
setting one of the submatrices to be the same as the
eavesdropper's channel matrix. Next, under the total power
constraint over all transmit antennas, we analytically solve the
channel-input covariance matrix optimization problem to fully
characterize the MIMOME secrecy capacity. Typically, the MIMOME
optimization problems are non-concave. However, thank to the
proposed degraded channel, we can transform the stochastic MIMOME
optimization problem to be a Schur-concave one and then find its
solution. Besides total power constraint, we also investigate the
secrecy capacity when the transmitter is subject to the practical
per-antenna power constraint. The corresponding optimization
problem is even more difficult since it is not Schuar-concave.
Under the two power constraints considered, the corresponding
MIMOME secrecy capacities can both scale with the signal-to-noise
ratios (SNR) when the difference between numbers of antennas at
legitimate receiver and eavesdropper are large enough. However,
when the legitimate receiver and eavesdropper have a single
antenna each, such SNR scalings do not exist for both cases.
\end{abstract}

\section{ Introduction}
Key-based enciphering is a well-adopted technique to ensure the
security in current data transmission system. However, for secure
communications in wireless networks, the distributions and
managements of secret keys may be challenging tasks
\cite{Liang_same_marginal}. The physical-layer security introduced
in \cite{csiszar1978broadcast}\cite{Wyner_wiretap} is appealing
due to its keyless nature. The basic building block of
physical-layer security is the so-called wiretap channel. In this
channel, a source node wants to transmit confidential messages
securely to a legitimate receiver and to keep the eavesdropper as
ignorant of the message as possible. Wyner \cite{Wyner_wiretap}
characterized the secrecy capacity of the discrete memoryless
wiretap channel, in which the secret key was not used. The secrecy
capacity is the largest secrecy rate of communication between the
source and the destination nodes with the eavesdropper knowing no
information of the messages. In order to meet the demands of high
data rate transmissions and improve the connectivities of the
secure networks \cite{Secureconnect}, the multiple antenna systems
with security concerns are considered by several authors. In
\cite{Shafiee_secrecy_2_2_1_J}, the secrecy capacity of a Gaussian
channel with two-input, two-output, single-antenna-eavesdropper
was first characterized. This result was extended by
\cite{Khisti_MIMOME}, in which the secrecy capacities of general
Gaussian multiple-input multiple-output
multiple-antenna-eavesdropper (MIMOME) channels are proved. In
wireless environments, with perfect knowledge of the legitimate
receiver's channel state information at the transmitter (CSIT),
the time-varying characteristic of fading eavesdropper channels
can be further exploited to enhance the secrecy
\cite{gopala2008secrecy,Li_fading_secrecy_j,CollingsSecrecy,Secrecy_scJSAC}.

However, to attain the secrecy capacity results in
\cite{Shafiee_secrecy_2_2_1_J,Khisti_MIMOME,gopala2008secrecy}, at
least the perfect knowledge of the legitimate receiver's CSIT is
required. For the fast fading channels, it may be hard to track
the rapidly varying channel coefficients because of the limited
feedback bandwidth and the delay caused by the channel estimation.
Thus for fast-fading channels, it is more practical to consider
the case with only partial CSIT of the legitimate channel. For the
setting where only statistical CSIT of both legitimate and
eavesdropper channels is available, the secrecy capacity is only
rigorously characterized for multiple-input single-output
single-antenna-eavesdropper (MISOSE) Rayleigh fast-faded channels
\cite{Secrecy_sc}. A negative phenomenon, revealed in
\cite{Secrecy_sc}, showed that for the MISOSE channels with only
statistical CSIT, the secrecy capacities would neither scale with
the number of transmit antennas nor the signal-to-noise ratio
(SNR). Thus using multiple transmitter antennas in the MISOSE
system limitedly helps increase the secrecy capacity, compared
with the system using single transmitter antenna.

In this paper, we want to overcome the aforementioned drawbacks of
the MISOSE system. We then focus on the MIMOME fast Rayleigh-faded
channels where the transmitter only have the statistical CSIT of
the legitimate and eavesdropper channels. Two different power
constraints are considered. One is the total power constraint over
all transmit antennas, and the other is the more practical
per-antenna power constraint. Under the total power constraint
over all transmit antennas, the MIMOME secrecy capacity is
characterized when antennas of the legitimate receiver are more
than (or equal to) those of the eavesdropper. To the best of the
author's knowledge, this is the first MIMOME secrecy capacity
result with partial CSIT. Compared with our previous works
\cite{Secrecy_sc}, new proof techniques are developed for the
MIMOME channels. First, we establish a new secrecy capacity
upper-bound by dividing the channel matrix of legitimate receiver
into two submatrices, one of the submatrices has dimensions equal
to those of the eavesdropper's channel matrix, to form an MIMOME
degraded channel. Second, instead of using completely monotone
property as \cite{Secrecy_sc}, which may only exists for the
MISOSE problem, we solve the stochastic MIMOME optimization
problem by transforming it to an equivalent Schur-concave problem.
The key to this transformation is using the proposed equivalent
MIMOME degraded channel. According to our secrecy capacity
results, on the contrary to \cite{Secrecy_sc}, we observe that the
SNR scaling of secrecy capacity can be obtained in the MIMOME
channels when the number of antenna of the legitimate receiver is
larger than (but not equal to) that of the eavesdropper. Also when
the number of transmit antennas is fixed, increasing the
difference between numbers of antennas of legitimate receiver and
eavesdropper helps to increase the secrecy capacities.

For the cases where the transmitters are subject to the practical
per-antenna power constraints, we also fully characterize the
secrecy capacities for the MISOSE wiretap channels with
statistical CSIT. Because the optimization problem subject to the
per-antenna power constraint is not Schur-concave, only numerical
algorithm is developed to find the secrecy capacity for the MIMOME
channel. However, we can still show that under this constraint,
when there are sufficient transmit antennas with large enough
transmitted power each, the MIMOME secrecy capacities can also
scale with SNRs. Such a scaling cannot be obtained for the MISOSE
wiretap channels under per-antenna power constraint, even when all
transmit antennas are allowed to transmit with large power.

Under total power constraint over all transmit antennas, several
works have studied the secrecy capacities in various channel
settings. For channels with full CSIT, the secrecy capacities were
found in \cite{Shafiee_secrecy_2_2_1_J,Khisti_MIMOME}; and for
ergodic slow fading channels with partial CSIT, they were found in
\cite{gopala2008secrecy}. However, for fast fading channels with
full legitimate CSIT and statistical eavesdropper CSIT, the
secrecy capacities are unknown. And several works instead studied
the achievable secrecy rate in these channels
\cite{Li_fading_secrecy_j,CollingsSecrecy,Secrecy_scJSAC}. The
settings in
\cite{Shafiee_secrecy_2_2_1_J,Khisti_MIMOME,gopala2008secrecy,Li_fading_secrecy_j,CollingsSecrecy,Secrecy_scJSAC}
are fundamentally different to ours. In this paper, fast fading
channels with only statistical CSIT of both legitimate receiver
and eavesdropper are considered. Our works are the first secrecy
capacity results for MIMOME channels with statistical CSIT. More
related works and the comparisons with our works can be found in
our MISOSE previous work \cite{Secrecy_sc} and references within.
Subject to the per-antenna power constraint, secrecy rate
optimization with full CSIT was studied in \cite{li2010optimal}.
Besides the difference between CSIT assumptions compared to ours,
the channel input matrix in \cite{li2010optimal} was optimized
numerically and the optimality was not guaranteed. However, in our
work, optimal channel input covariance matrix for the MISOSE
channel is analytically solved, and our numerical algorithm for
the MIMOME channel can guarantee the optimality.

The rest of the paper is organized as follows. In Section
\ref{Sec_system_model} we introduce the considered system model.
In Section \ref{Sec_total_power}, under total power constraint
over all tranmitter antennas, we prove the secrecy capacities for
fast fading MIMOME wiretap channels with statistical CSIT. In
Section \ref{Sec_per_antenna}, we investigate secrecy capacities
under per-antenna power constraints. Simulation results are
provided in Section \ref{Sec_simulation}, and the conclusion is
given in Section \ref{Sec_conclusion}.

\underline{\bf{Notations}}: In this paper, lower and upper case
bold alphabets denote vectors and matrices, respectively. The
zero-mean complex Gaussian random vector with covariance matrix
$\Sigma$ is denoted as $CN(0, \Sigma)$. The mutual information
between two random vectors $\bx$ and $\by$ is $I(\bx;\by)$, while
the conditional differential entropy is $h(\bx|\by)$. The
superscript $(.)^H$ denotes the transpose complex conjugate.
$|\mathbf{A}|$ and $|a|$ represent the determinant of the square
matrix $\mathbf{A}$ and the absolute value of the scalar variable
$a$, respectively. The trace of $\mathbf{A}$ is denoted by
$\tr(\mathbf{A})$. The element of $\bA$ in the $i$th row and $j$th
column is $\{\bA\}_{i,j}$. A diagonal matrix whose diagonal
entries are $a_1 \ldots a_k$ is denoted by $diag(a_1 \ldots a_k)$.
The positive semidefinite ordering between Hermitian matrices
$\bA$ and $\bB$ are denoted by $\bA\succeq \bB$ ($\bA\succ \bB$),
where $\bA-\bB$ is a positive semi-definite (definite) matrix.


\section{System Model}\label{Sec_system_model}
In the considered MIMOME wiretap channel, we study the problem of
reliably communicating a secret message $w$ from the transmitter
to the legitimate receiver, subject to a constraint on the
information attainable by the eavesdropper (in upcoming
\eqref{eq_equivocation_given_h}). The transmitter has $n_t$
antennas, while the legitimate receiver and eavesdropper
respectively have $n_r$ and $n_e$ antennas as
\begin{align}
\by&=\bH\mathbf{x}+\bn_{y}, \label{EQ_main_MISOSE} \\
\bz&=\bG\mathbf{x}+\bn_{z},\label{EQ_Eve_MISOSE}
\end{align}
where $\mathbf{x} \in \mathds{C}^{n_t \times 1}$ represents the
transmitted vector signal; the legitimate channel matrix is $\bH
\in \mathds{C}^{n_r \times n_t}$ while the eavesdropper channel
matrix is $\bG \in \mathds{C}^{n_e \times n_t}$; $\bn_y$ and
$\bn_z$ are additive white Gaussian noise vectors at the
legitimate receiver and eavesdropper, respectively, with each
element independent and identically distributed (i.i.d.),
circularly symmetric, and having zero mean and unit variance. The
channels are assumed to be fast Rayleigh fading, that is, each
element of $\bH$ and $\bG$ is i.i.d distributed as
\begin{equation} \label{eq_Rayleigh}
CN(0, \sigma_{\bh}^2) \;\; \mbox{and} \;\; CN(0, \sigma_{\bg}^2),
\end{equation}
respectively, while the channel coefficients change in each symbol
time. The $\mathbf{H}$, $\mathbf{G}$, $\bn_y$ and $\bn_z$ are
independent. We assume that the legitimate receiver knows the
instantaneous channel state information of $\mathbf{H}$ perfectly,
while the eavesdropper knows those of $\mathbf{H}$ and
$\mathbf{G}$ perfectly. As for the CSIT, only the distributions of
$\mathbf{H}$ and $\mathbf{G}$ are known while the realizations of
$\mathbf{H}$ and $\mathbf{G}$ are unknown at the transmitter.

In this work, we consider two kinds of constraints for the channel
input $\bx$. The first one is the total power constraint over all
transmitter antennas  as
\begin{equation} \label{power_cons}
\mathrm{Tr}(\Sigma_{\bx}) \leq P,
\end{equation}
where $\Sigma_{\bx}$ is the covariance matrix of $\mathbf{x}$.
Note that under \eqref{power_cons}, the transmitter can perform
power allocation between transmitter antennas to increase the
secrecy capacity. In addition to \eqref{power_cons}, we also
consider a more practical per-antenna constraint for the
transmitter antennas as
\begin{equation} \label{eq_idv_power_cons}
\{\Sigma_\bx\}_{ii}\leq P_i,
\end{equation}
for $i=1,\ldots, n_t$. Note that the per-antenna constraint
\eqref{eq_idv_power_cons} is more stringent than the total power
constraint \eqref{power_cons} when $\Sigma_{i=1}^{N_t} P_i \leq
P$.

The perfect secrecy and the corresponding secrecy capacity are
defined as follows. Consider a $(2^{NR}, N)$-code with an encoder
that maps the message $w\in \calW_N=\{1,2,\ldots, 2^{NR}\}$ into a
length-$N$ codeword, and a decoder at the legitimate receiver that
maps the received sequence $y^N$ (the collections of $y$ over the
code length $N$) from the legitimate channel
\eqref{EQ_main_MISOSE} to an estimated message $\hat w\in\calW_N$.
We then have the following definitions, where $\bz^N$,
$\mathbf{H}^N$, and $\mathbf{G}^N$ are the collections of $\bz$,
$\mathbf{H}$, and $\mathbf{G}$ over the code length $N$,
respectively.
\begin{Def}[Secrecy Capacity
\cite{Liang_same_marginal}\cite{Wyner_wiretap}\cite{gopala2008secrecy}]
\label{Def_Perfect} {\it Perfect secrecy is achievable with rate
$R$ if, for any $\varepsilon>0$, there exists a sequence of
$(2^{NR}, N)$-codes and an integer $N_0$ such that for any $N>N_0$
\begin{align}
&{R_e=h(w|\bz^N,\mathbf{H}^N, \mathbf{G}^N)/N \geq R-\varepsilon},
\label{eq_equivocation_given_h} \\ \mathrm{and} \;\;\; &{\rm
Pr}(\hat{w}\neq w) \leq \varepsilon, \notag
\end{align}
where { $R_e$ in \eqref{eq_equivocation_given_h} is the
equivocation} rate and $w$ is the secret message. The {\bf secrecy
capacity} is the supremum of all achievable secrecy rates.}
\end{Def}
Note that the perfect secrecy requirement in
\eqref{eq_equivocation_given_h} is measured by $I(w;
\bz^N,\mathbf{H}^N, \mathbf{G}^N)/N=R-R_e$, which is based on all
the observations $(\bz^N,\mathbf{H}^N, \mathbf{G}^N)$ the
eavesdropper has.

From Csisz{\'{a}}r and K{\"{o}}rner's seminal work
\cite{csiszar1978broadcast}, we know that the secrecy capacity of
MIMOME channel \eqref{EQ_main_MISOSE} and \eqref{EQ_Eve_MISOSE} is
\begin{align}
&\underset{U}{\max}\;\;I(U;\by,\bH)-I(U;\bz,\bH,\bG), \notag \\
=&\underset{U}{\max}\;\;I(U;\by|\bH)-I(U;\bz|\bH,\bG),
\label{EQ_partial_CSI_rate}
\end{align}
where $U$ is an auxiliary random variable satisfying the Markov
relationship $U\rightarrow \bx \rightarrow (\by,\bH), (\bz,
\bH,\bG)$, and \eqref{EQ_partial_CSI_rate} results from the fact
that the transmitter does not have the knowledge of the
realizations of $\bH$ and $\bG$. However, the optimal choice of
$U$ which maximizes the secrecy capacity of considered fast fading
MIMOME channel is \textit{unknown}. In this paper, we want to
fully characterize the optimal $U$. Note that the secrecy capacity
considered here is achieved by encoding over multiple channel
states and the perfect secrecy constraint must be satisfied for
all $N>N_0$. This implies that no secrecy outage
\cite{yuksel2011diversity} \cite{CollingsSecrecy} is allowed. In
delay-limited applications, such perfect secrecy condition may not
be achievable and, thus, a tradeoff exists between secrecy rate
and secrecy outage probability. These issues have been discussed
in \cite{yuksel2011diversity} \cite{CollingsSecrecy} and are
beyond the scope of this paper.

\section{Secrecy capacity under the \\ total power constraint}\label{Sec_total_power}
In this section, we explicitly find the optimal $U$ in
\eqref{EQ_partial_CSI_rate} for wiretap channel
\eqref{EQ_main_MISOSE}\eqref{EQ_Eve_MISOSE}, and fully
characterize the MIMOME secrecy capacity with statistical CSIT in
the upcoming Theorem \ref{Theorem_No_CQI}. When there is full CSIT
\cite{Khisti_MIMOME}, one can find the optimal auxiliary random
variable $U$ by constructing an equivalent degraded MIMOME channel
to upper-bound the secrecy capacity. However, with only
statistical CSIT, if one naively applies the degraded channel
construction method in \cite{Khisti_MIMOME}, the resulting secrecy
capacity upper bound will depend on the realizations of
($\bH$,$\bG$) and become very loose. Thus in general it is very
hard to find the optimal $U$ maximizing
\eqref{EQ_partial_CSI_rate}

Due to the difficulty mentioned in the previous paragraph, with
statistical CSIT, finding the optimal $U$ maximizing
\eqref{EQ_partial_CSI_rate} is very hard in general. However, in
the following Lemma, for the special cases where the MIMOME
channels are Rayleigh faded as \eqref{eq_Rayleigh}, we show that
one may still construct a degraded channel for upper-bounding and
find the optimal $U$ maximizing \eqref{EQ_partial_CSI_rate}. The
key for building this equivalent degraded MIMOME channel for
\eqref{EQ_main_MISOSE}\eqref{EQ_Eve_MISOSE} is replacing the
legitimate channel $\bH$ with equivalent $\bH'$ in upcoming
\eqref{eq_H_divide} as follows. When $n_r \geq n_e$, one can
divide the legitimate channel matrix $\bH$ in
\eqref{EQ_main_MISOSE} as two submatrices
\begin{equation} \label{eq_split_H}
\bH=[\bH^T_{(n_r-n_e)} \; \bH^T_{n_e}]^T,
\end{equation}
where $\bH_{(n_r-n_e)} \in \mathds{C}^{(n_r-n_e) \times n_t}$ and
$\bH_{n_e} \in \mathds{C}^{n_e \times n_t}$, with each element of
$\bH_{(n_r-n_e)}$ and $\bH_{n_e}$ distributed as i.i.d. Gaussian
$CN(0,\sigma_h^2)$. From the properties of complex Gaussian
distributions, the distribution of $\bH$ is the same as that of
\begin{equation} \label{eq_H_divide}
\bH'=[\bH^T_{(n_r-n_e)} \;
\left(\frac{\sigma_{\bh}}{\sigma_{\bg}}\right)\bG^T]^T,
\end{equation}
because each element of eavesdropper channel matrix $\bG$ is
distributed as $CN(0,\sigma_g^2)$ according to
\eqref{eq_Rayleigh}. With \eqref{eq_H_divide}, one can build an
equivalent degraded MIMOME channel for
\eqref{EQ_main_MISOSE}\eqref{EQ_Eve_MISOSE} (see upcoming
\eqref{eq_same_dist_leg_channel} \eqref{eq_degraded_markov}),
where the received signal at eavesdropper can be treated as a
degraded version of that at the legitimate receiver. Then we can
obtain a tight secrecy capacity upper bound and have
the following property of the optimal $U$ for \eqref{EQ_partial_CSI_rate}. \\
\begin{lemma} \label{Lem_No_CQI}
For the MIMOME fast Rayleigh fading wiretap channel
\eqref{EQ_main_MISOSE}\eqref{EQ_Eve_MISOSE} with the statistical
CSIT of $\mathbf{H}$ and $\mathbf{G}$, using Gaussian $\bx$
without prefixing $U \equiv \bx$ is the optimal transmission
strategy for \eqref{EQ_partial_CSI_rate} under \eqref{power_cons}
when $n_r \geq n_e$ and $\sigma_h \geq \sigma_g$, where $n_r$ and
$n_e$ respectively are the number of antennas at the legitimate
receiver and eavesdropper, while $\sigma_h$ and $\sigma_g$
respectively are the variances of the legitimate and eavesdropper
channels as
\eqref{eq_Rayleigh}.\\
\end{lemma}

\begin{proof}
We first form the degraded MIMOME channel with respect to
\eqref{EQ_main_MISOSE}\eqref{EQ_Eve_MISOSE}. Here we will only
consider the cases where $\sigma^2_g>0$, since the case with
$\sigma^2_g=0$ is a trivial one which corresponds to channel
without eavesdropper. Note that $\bH$ in \eqref{eq_split_H} has
the same distribution as $\bH'$ in \eqref{eq_H_divide}. Then we
can form an equivalent received signal
\begin{equation} \label{eq_same_dist_leg_channel}
\by'=\bH'\bx+\bn_y,
\end{equation}
which has the same marginal distribution as the legitimate
received signal $\by$ in \eqref{EQ_main_MISOSE}. Now, rewrite the
above $\by'$ using \eqref{eq_H_divide} as
\begin{align}
\by'&=[\bH_{(n_r-n_e)}^T \;
\left(\frac{\sigma_{\bh}}{\sigma_{\bg}}\right)\bG^T]^T\mathbf{x}+\bn_{y}\notag
\\ &=\left[\left(\by'_{(n_r-n_e)}\right)^T\;\left(\by'_{n_e}\right)^T\right]^T, \label{eq_same_dist_channel_nor}
\end{align}
where $\by'_{(n_r-n_e)} \in \mathds{C}^{(n_r-n_e) \times 1}$ and
$\by'_{n_e} \in \mathds{C}^{n_e \times 1}$. Now we consider
\[
\frac{\sigma_{\bg}}{\sigma_{\bh}}\by'=\left[\frac{\sigma_{\bg}}{\sigma_{\bh}}\left(\by'_{(n_r-n_e)}\right)^T\;\left(\bG\bx+(\sigma_{\bg}/\sigma_{\bh})\bn_{y,n_e}\right)^T\right]^T,
\notag
\]
where the noise vector $\bn_{y,n_e} \in \mathds{C}^{n_e \times 1}$
comes from dividing the noise vector at the legitimate receiver as
$\bn_y=[\bn^T_{y,(n_r-n_e)} \;\; \bn^T_{y,n_e}]^T$. From
\eqref{EQ_Eve_MISOSE}, it is clear that when $\sigma_{\bh} \geq
\sigma_{\bg}$, we have the markov relationship that given $\bH'$
\begin{equation} \label{eq_degraded_markov}
\bx \rightarrow \frac{\sigma_{\bg}}{\sigma_{\bh}}\by' \rightarrow
\bz.
\end{equation}
The degraded MIMOME wiretap channel
$(\bx,(\sigma_{\bg}/\sigma_{\bh})\by', \bz)$ is then formed.

Now based on the proposed degraded channel
$(\bx,(\sigma_{\bg}/\sigma_{\bh})\by', \bz)$, we know that the
secrecy capacity is upper-bounded by
\begin{align}
C^{\;t}_s &\leq \underset{\bx}{\max} \;
I(\bx;\frac{\sigma_{\bg}}{\sigma_{\bh}}\by'|\bH')-I(\bx;\bz|\bG),
\label{eq_MIMOME_capacity_I_a} \\
&= \underset{\bx}{\max} \; I(\bx;\by|\bH)-I(\bx;\bz|\bG),
\label{eq_MIMOME_capacity_I}
\end{align}
where the inequality follows \cite{Secrecy_sc}, with the
right-hand-side (RHS) being the secrecy capacity of our degraded
channel but forcing the eavesdropper knowing $(\bz,\bG)$ instead
of $(\bz,\bH',\bG)$; and the equality comes from that
$(\by',\bH')$ and $(\by,\bH)$ have the same distributions. From
\cite{csiszar1978broadcast}, we also know that the RHS of
\eqref{eq_MIMOME_capacity_I} is achievable. Thus the RHS of
\eqref{eq_MIMOME_capacity_I} is the secrecy capacity $C^{\;t}_s$.
Furthermore, from \cite{Khisti_MIMOME}, we know that Gaussian
$\bx$ is optimal for the secrecy capacity in
\eqref{eq_MIMOME_capacity_I} under total power constraint
\eqref{power_cons}. Then our claim is valid.
\end{proof}

Note that on the contrary to the upper-bound with full CSIT in
\cite{Khisti_MIMOME}, our secrecy capacity upper bound
\eqref{eq_MIMOME_capacity_I_a} is independent of the realizations
of $\bH$ and $\bG$. This is why our bounds are tight for MIMOME
channels with only statistical CSIT. Compared to the proof for the
MISOSE secrecy capacity \cite{Secrecy_sc}, the key for deriving
the tight MIMOME upper-bound is separating the legitimate channel
matrix $\bH$ by two submatrices $\bH_{(n_r-n_e)}$ and $\bH_{n_e}$
as \eqref{eq_split_H}, and only introducing correlations between
$\bH_{n_e}$ and $\bG$ as \eqref{eq_H_divide}. One can treat the
unchanged submatrix $\bH_{(n_r-n_e)}$ as a ``safe'' channel matrix
without being eavesdropped, which provides SNR scaling for the
secrecy capacity. In MISOSE channel, such a $\bH_{(n_r-n_e)}$ does
not exist because $n_r=n_e=1$, and there is no SNR scaling. This
intuition is verified rigorously from the upcoming Theorem
\ref{Theorem_No_CQI} and Corollary \ref{coro_total_SNR}.

Now we fully characterize the MIMOME secrecy capacity based on
Lemma \ref{Lem_No_CQI}. Typically, the MIMOME secrecy capacity
optimization problems like the upcoming \eqref{eq_MISOSE_capacity}
are non-concave. This is due to that the MIMOME secrecy
capacities, such as \eqref{eq_MISOSE_capacity}, are a difference
of two concave functions. However, with aids of the degraded
MIMOME channels formed by \eqref{eq_H_divide}, the stochastic
MIMOME optimization problem \eqref{eq_MISOSE_capacity} can be
transformed to be a Schur-concave problem. This transformation
helps a lot to find the optimal solution
\eqref{eq_total_cons_optimal}. Note that the completely monotone
property for MISOSE optimization problem \cite{Secrecy_sc} may not
exist for the MIMOME one, and thus the method in \cite{Secrecy_sc}
is
hard to be extended to the MIMOME cases. \\
\begin{Theo} \label{Theorem_No_CQI}
Under the total power constraint
\begin{equation} \label{power_cons1}
\mathrm{Tr}\left(\Sigma_{\bx}\right) \leq P,
\end{equation}
the MIMOME secrecy capacity $C^{\;t}_s$ when $n_r \geq n_e$ and
$\sigma_h \geq \sigma_g$ is
\begin{equation}
\max_{\Sigma_{\bx}}\!\left(\E_{\bH}\!\left[\!\log\left|\bI+\bH\Sigma_{\bx}\bH^\dag\!\right|\!\right]\!-\!\E_{\bG}\!\left[\!\log\!\left|\bI+\bG\Sigma_{\bx}\bG^\dag\!\right|\!\right]\!\right),
\label{eq_MISOSE_capacity}
\end{equation}
and the optimal channel input covariance matrix subject to
\eqref{power_cons1} is
\begin{equation} \label{eq_total_cons_optimal}
{\Sigma}^*_{\bx}=\frac{P}{n_t}\bI.
\end{equation}
\[
\]
\end{Theo}

\begin{proof}
First we show that under our setting, subject to
\eqref{power_cons1}, the stochastic MIMOME optimization problem
\eqref{eq_MISOSE_capacity} can be transformed to an equivalent
concave problem. The key is cleverly using the same marginal
channel formed by \eqref{eq_H_divide}. Note that the objective
function in \eqref{eq_MISOSE_capacity} can be rewritten as
$I(\bx;\by|\bH)-I(\bx;\bz|\bG)$ with $\bx \sim CN(0,\Sigma_{\bx})$
as
\begin{align}
I(\bx;\by|\bH)-I(\bx;\bz|\bG) \overset{(a)}=&
I(\bx;\frac{\sigma_{\bg}}{\sigma_{\bh}}\by'|\bH')-I(\bx;\bz|\bG),
\notag \\
\overset{(b)}=&
I(\bx;\frac{\sigma_{\bg}}{\sigma_{\bh}}\by'|\bz,\bH'),
\label{eq_Convex}
\end{align}
and then from \cite{Khisti_MIMOME}, we know that the RHS of
(\ref{eq_Convex} b) is concave in $\Sigma_{\bx}$. The
(\ref{eq_Convex} a) comes from \eqref{eq_MIMOME_capacity_I}, while
(\ref{eq_Convex} b) comes from \eqref{eq_H_divide} and the fact
that
\begin{align}
&I(\bx;\frac{\sigma_{\bg}}{\sigma_{\bh}}\by'|\bH')-I(\bx;\bz|\bG)
\notag \\
\overset{(a)}=&
I(\bx;\frac{\sigma_{\bg}}{\sigma_{\bh}}\by',\bz|\bH_{(n_r-n_e)},\bG)-I(\bx;\bz|\bH_{(n_r-n_e)},\bG),
\notag \\
\overset{(b)}=&
I(\bx;\frac{\sigma_{\bg}}{\sigma_{\bh}}\by'|\bz,\bH_{(n_r-n_e)},\bG).
\label{eq_Convex1}
\end{align}
In the above, (\ref{eq_Convex1} a) comes from that
\begin{align}
&I(\bx;\frac{\sigma_{\bg}}{\sigma_{\bh}}\by'|\bH')=I(\bx;\frac{\sigma_{\bg}}{\sigma_{\bh}}\by',\bz|\bH_{(n_r-n_e)},\bG),
\notag \\ \mbox{and} \;\;\;
&I(\bx;\bz|\bG)=I(\bx;\bz|\bH_{(n_r-n_e)},\bG) \notag,
\end{align}
with the former resulting from \eqref{eq_H_divide} and the Markov
relationship \eqref{eq_degraded_markov} and the latter resulting
from the independence of $\bH_{(n_r-n_e)}$ with $(\bx,\bz,\bG)$;
while (\ref{eq_Convex1} b) comes from the chain rule of the mutual
information \cite{Book_Cover}.

After showing that \eqref{eq_MISOSE_capacity} can be transformed
to an equivalent concave problem, now we show that
\eqref{eq_MISOSE_capacity} can be further transformed to a
symmetric optimization problem. Then we can explore the
Schuar-concavity of the secrecy capacity
\eqref{eq_MISOSE_capacity} to show that the optimal
$\Sigma_{\bx}=\alpha\bI$ where $0 \leq \alpha \leq P/n_t$. For any
non-diagonal ${\Sigma}_{\bx}$, we can apply the eigenvalue
decomposition on it as ${\Sigma}_{\bx}=\bU\bD\bU^\dag$, where
$\bU$ is unitary and $\bD$ is diagonal. Then for the objective
function in \eqref{eq_MISOSE_capacity}, setting
${\Sigma}_{\bx}=\bU\bD\bU^\dag$ and ${\Sigma}_{\bx}=\bD$ will
result in the same value since $\bH$ and $\bG$ are Rayleigh
distributed as \eqref{eq_Rayleigh}. Also the total power
constraint \eqref{power_cons1} can be transformed as
\begin{equation} \label{eq_evd_cons}
\mathrm{Tr}\left(\bU\bD\bU^\dag\right)=\mathrm{Tr}\left(\bU^\dag\bU\bD\right)=\mathrm{Tr}\left(\bD\right)
\leq P.
\end{equation}
In the following, we then focus on the following optimization
problem instead
\begin{equation}
\max_{\bD}\!\left(\E_{\bH}\!\left[\!\log\left|\bI+\bH\bD\bH^\dag\!\right|\!\right]\!-\!\E_{\bG}\!\left[\!\log\!\left|\bI+\bG\bD\bG^\dag\!\right|\!\right]\!\right),
\label{eq_MISOSE_capacityD}
\end{equation}
where the diagonal matrix $\bD$ satisfying \eqref{eq_evd_cons}.
Because $\bH$ and $\bG$ are Rayleigh faded as \eqref{eq_Rayleigh},
\eqref{eq_MISOSE_capacityD} is symmetric for $\bD$. That is, all
permutations of the diagonal terms of $\bD$ will result in the
same value for the objective function in
\eqref{eq_MISOSE_capacityD}. Together with the fact that
\eqref{eq_MISOSE_capacityD} is concave as shown previously, we
know that the objective function \eqref{eq_MISOSE_capacityD} is
Schur-concave. Subject to the constraint $\mathrm{Tr}(\bD) \leq
P$, from \cite{Marshall_Inequalities}, it can be easily shown that
the optimal $\bD$ (and thus $\Sigma_{\bx}$) is $\alpha \bI$ where
$0 \leq \alpha \leq P/n_t$. This result comes from the fact that
the diagonal entries $\{\alpha,\ldots,\alpha\}$ are majorized by
any other diagonal entries $\{d_1,\ldots,d_{n_t}\}$ of $\bD$ if
$\frac{1}{n_t}\sum_{i=1}^{n_t}{d_i}=\alpha$, and properties of the
Schur-concave function \cite{Marshall_Inequalities}.

Finally, we will show that using all available power, that is,
$\Sigma_{\bx}=(P/n_t)\bI$ is optimal for
\eqref{eq_MISOSE_capacity}. To do this, we will first show an
important property of the objective function in
\eqref{eq_MISOSE_capacity} with respect to the matrix partial
ordering as in the following Lemma.
\\

\begin{lemma} \label{lem_mono}
For the objective function of the optimization problem in
\eqref{eq_MISOSE_capacity},
\begin{equation} \label{eq_capacity_object}
R_s(\Sigma_{\bx}) \triangleq
\E_{\bH}\!\left[\!\log\left|\bI+\bH\Sigma_{x}\bH^\dag\!\right|\!\right]\!-\!\E_{\bG}\!\left[\!\log\!\left|\bI+\bG\Sigma_{x}\bG^\dag\!\right|\!\right],
\end{equation}
we have the following properties when $n_r \geq n_e$ and
$\sigma_h \geq \sigma_g$ \\

\noindent (I) If two channel input covariance matrices satisfy
$\Sigma^2_x \succeq \Sigma^1_x \succeq \mathbf{0}$, then
\begin{equation} \label{eq_mono_sigma_x}
R_s(\Sigma^2_x)\geq R_s(\Sigma^1_x).
\end{equation} \\
(II) If $\Sigma^1_x \succ \mathbf{0}$, then $R_s(\Sigma^1_x)>0$.
\[
\]
\end{lemma}
The proof of Lemma \ref{lem_mono} is given in Appendix
\ref{app_mono}, which results from the equivalent degraded MIMOME
channels described in the proof of Lemma \ref{Lem_No_CQI}. As
aforementioned, we know the optimal $\Sigma_{\bx}=\alpha \bI$ for
\eqref{eq_MISOSE_capacity} subject to \eqref{power_cons1}. Now we
can show that $\alpha=P/n_t$ is optimal for
\eqref{eq_MISOSE_capacity} using Lemma \ref{lem_mono} as follows.
Let us substitute $\Sigma_{\bx}=\alpha\bI$ into the objective
function in \eqref{eq_MISOSE_capacity}, it becomes
\begin{equation} \label{eq_total_alpha}
\E_{\bH}\!\left[\!\log\left|\bI+\alpha\bH\bH^\dag\!\right|\!\right]\!-\!\E_{\bG}\!\left[\!\log\!\left|\bI+\alpha\bG\bG^\dag\!\right|\!\right]=C(\alpha),
\end{equation}
where $0 \leq \alpha \leq P/n_t$. If $\alpha>0$, we know that
\[
(P/n_t)\bI \succeq \alpha \bI \succ \mathbf{0}.
\]
Then from Property (I) of Lemma \ref{lem_mono}, we know that
$C(P/n_t) \geq C(\alpha)$ for any $\alpha>0$. As for the case
$\alpha = 0$, we know that $C(0)=0$ from \eqref{eq_total_alpha}.
However, if $\alpha>0$, we have $ C(\alpha)>0$ from Property (II)
of Lemma \ref{lem_mono}. Then we know that
\[
C(P/n_t) \geq C(\alpha)> C(0),
\]
for any $\alpha>0$. Thus $\alpha=P/n_t$ is optimal for
\eqref{eq_total_alpha}. It concludes our proof.
\end{proof}

\textit{Remark:} Note that in the final step of Theorem
\ref{Theorem_No_CQI}'s proof, we show that the optimal channel
input covariance matrix happens when the total power constraint in
\eqref{power_cons1} is met with equality. This fact is not always
true for the wiretap channel and counter examples are given
\cite{Secrecy_Finite}\cite{Li_fading_secrecy_j}. When the transmit
power increases, both the SNRs at the legitimate receiver and the
eavesdropper increase. Then the secrecy capacity may not always be
maximized with using all available power. Indeed, this property
was also examined in \cite{Secrecy_scJSAC}. However, in
\cite{Secrecy_scJSAC}, the perfect CSIT of the legitimate receiver
is assumed to be  available, which is fundamentally different to
our setting.

Now we have the following result which characterizes the secrecy
capacity in Theorem \ref{Theorem_No_CQI} with respect to the
number of antennas $(n_t,n_r,n_e)$ at the transmitter, legitimate
receiver and eavesdropper respectively. The proof is given in
Appendix
\ref{App_coro_total_SNR}. \\

\begin{coro} \label{coro_total_SNR}
Under total power constraint \eqref{power_cons1}, we have the
following asymptotic results for the MIMOME secrecy capacity
$C^{\;t}_s$ as
\[
\lim_{P_g \rightarrow
\infty}\frac{C^{\;t}_s-n_r\log(\sigma^2_h/\sigma^2_g)}{\log
P_g}=\left\{\begin{array}{lc} n_r-n_e, & n_t \geq n_r > n_e\\
n_t-n_e, & n_r \geq n_t > n_e\\
0, &   n_r \geq n_e \geq n_t\\
\end{array}\right.
\]
when $\sigma_h \geq \sigma_g > 0$, where $P_g=\sigma^2_gP$ is the
equivalent received SNR at the eavesdropper. \\
\end{coro}

In \cite{Secrecy_sc}, it was shown that when $n_r=n_e=1$, the
secrecy capacity does not scale with $P$. Our Corollary
\ref{coro_total_SNR} shows that to make secrecy capacity scale
with $P$, we must let $n_r-n_e$ larger than (but not equal to)
zero. Also adding enough number of transmit antennas to make
$n_t>n_e$ is very important for increasing the secrecy capacity.
Indeed, from our numerical results in Section
\ref{Sec_simulation}, with fixed $n_t>n_e$, increasing the
difference $n_r-n_e$ will help to increase the secrecy capacities
for all SNR regimes (besides the high SNR regime proved in
Corollary \ref{coro_total_SNR}).

From Property (II) of Lemma \ref{lem_mono}, the secrecy capacity
in Theorem \ref{Theorem_No_CQI} is always positive when $P>0$,
$n_r \geq n_e$ and $\sigma_h > \sigma_g$. It will also be
interesting to see when the secrecy capacity is zero. We have the
following result, where the proof is similar
to that of Lemma \ref{Lem_No_CQI} and neglected.\\

\begin{coro}
\label{Coro_No_CQI_zero} For the MIMOME fast Rayleigh fading
wiretap channel \eqref{EQ_main_MISOSE}\eqref{EQ_Eve_MISOSE} with
the statistical CSIT of $\mathbf{H}$ and $\mathbf{G}$, the secrecy
capacity is zero when $n_r \leq n_e$ and $\sigma_h \leq \sigma_g$.
\end{coro}

\section{Secrecy capacity under the \\ per-antenna power constraint}\label{Sec_per_antenna}

After investigating the MIMOME secrecy capacity $C^{\;t}_s$ under
the total power constraint over all transmitter antennas, now we
consider the MIMOME channel under the more practical per-antenna
power constraint. First, we show that the MIMOME secrecy capacity
under the per-antenna power constraint can be expressed as the
following convex optimization problem, with proof given in
Appendix
\ref{app_Pro_MIMO_per_antenna}. \\

\begin{proposition} \label{Pro_MIMO_per_antenna}
Under the per-antenna power constraint
\begin{equation} \label{eq_per_power_cons}
\{\Sigma_{\bx}\}_{ii} \leq P_i,
\end{equation}
$i=1,\ldots,n_t$, when $n_r \geq n_e$ and $\sigma_h \geq
\sigma_g>0$, the MIMOME secrecy capacity $C^p_s$ is
\begin{equation} \label{eq_per_MIMOME_capacity_old}
\max_{\Sigma_{\bx}}\!\left(\E_{\bH}\!\left[\!\log\left|\bI+\bH\Sigma_{\bx}\bH^\dag\!\right|\!\right]\!-\!\E_{\bG}\!\left[\!\log\!\left|\bI+\bG\Sigma_{\bx}\bG^\dag\!\right|\!\right]\!\right),
\end{equation}
which can be transformed to a concave optimization problem subject
to \eqref{eq_per_power_cons} as
\begin{align}
\max_{\Sigma_{\bx}}\Bigg(&\E_{\bG,\bH_{(n_r-n_e)}}\!\left[\!\log\left|\bI+\Sigma_{\bx}\left(\bH^\dag_{(n_r-n_e)}\bH_{(n_r-n_e)}+\frac{\sigma^2_h}{\sigma^2_g}\bG^\dag\bG\right)\!\right|\right]
\notag
\\&-\!\E_{\bG}\!\left[\!\log\!\left|\bI+\bG\Sigma_{\bx}\bG^\dag\!\right|\!\right]
\Bigg),
 \label{eq_per_MIMOME_capacity}
\end{align}
where $\bG$ is the eavesdropper channel matrix and
$\bH_{(n_r-n_e)}$ is the sub-matrix of the legitimate channel as
defined in \eqref{eq_H_divide}. \\
\end{proposition}

Unlike the optimization in Theorem \ref{Theorem_No_CQI}, subject
to \eqref{eq_per_power_cons}, the optimal $\Sigma_{\bx}$ of
problem \eqref{eq_per_MIMOME_capacity_old} is very hard to find
analytically for general $n_r$. The main difficulty is that the
per-antenna power constraint \eqref{eq_per_power_cons} is not
symmetric. Note that \eqref{eq_per_power_cons} is related to
eigenvectors of $\Sigma_\bx$, while the total power constraint
\eqref{power_cons1} ($\mathrm{Tr}\{\Sigma_\bx\}$) is independent
of these eigenvectors. Thus although the total power constraint
can be transformed to \eqref{eq_evd_cons} via the eigenvalue
decomposition of $\Sigma_{\bx}$, the same technique cannot be
applied to the per-antenna constraint \eqref{power_cons1}. Thus we
only claim that the problem \eqref{eq_per_MIMOME_capacity_old}
subject to \eqref{eq_per_power_cons} can be transformed to a
concave one as \eqref{eq_per_MIMOME_capacity}, but not the
Schuar-concave one as in the proof of Theorem
\ref{Theorem_No_CQI}. Lack of symmetry in constraint
\eqref{eq_per_power_cons} makes finding analytical solution of it
difficult. Nevertheless, we still have the following result for
the structure of the optimal
$\Sigma^*_x$ in Proposition \ref{Pro_MIMO_per_antenna} as \\

\begin{proposition}
\label{Pro_Full_power} Under the per-antenna power constraint
$\{\Sigma_{\bx}\}_{ii} \leq P_i, i=1,\ldots,n_t$, when $n_r \geq
n_e$ and $\sigma_h \geq \sigma_g$, there exists an optimal channel
input covariance matrix $\Sigma^*_x$ for the MIMOME secrecy
capacity optimization problem \eqref{eq_per_MIMOME_capacity_old}
which satisfies
\begin{equation} \label{eq_per_equality}
\{\Sigma^*_{\bx}\}_{ii}=P_i, \;\;\; i=1,\ldots,n_t,
\end{equation}
that is, each antenna should fully use its available power. \\
\end{proposition}

\begin{proof}
We will prove it by showing that if there exists an optimal
$\Sigma_{\bx}$ which does not meet \eqref{eq_per_equality}, one
can find another $\Sigma^*_\bx$ meeting \eqref{eq_per_equality}
and having the same optimal value. Suppose that there exists an
optimal $\Sigma_{\bx}$ maximizing
\eqref{eq_per_MIMOME_capacity_old} which has some diagonal terms
using powers smaller than their maximum available powers. We
collect the indexes corresponding to these diagonal terms as a
non-empty set $\mathds{I}_{s} \neq \varnothing$, i.e.,
$\{\Sigma_{\bx}\}_{ii}<P_i,$ if $ i \in \mathcal{I}_{s} \subseteq
\{1,\ldots, n_t\}$. Then we can form another $\Sigma^*_\bx$ as
\[
\{\Sigma^*_{\bx}\}_{ij}=\left\{
\begin{array}{ll}
P_i & \mbox{if $i=j$ and $i \in \mathds{I}_{s} $,} \\
\{\Sigma_{\bx}\}_{ij} & \mbox{otherwise.}
\end{array} \right.
\]
Then we know that $\Sigma^*_{\bx}-\Sigma_{\bx} \succ 0$, because
$\Sigma^*_{\bx}-\Sigma_{\bx}$ is a diagonal matrix with
$|\mathds{I}_{s}|$ positive diagonal terms, while the other
diagonal terms are zero. Note that the size of set
$\mathds{I}_{s}$ meets $|\mathds{I}_{s}|>0$ and at least one
diagonal terms of $\Sigma^*_{\bx}-\Sigma_{\bx}$ is positive. Then
from Property (I) of Lemma \ref{lem_mono} we know that for the
objective function $R_s(.)$ in \eqref{eq_per_MIMOME_capacity_old},
$R_s(\Sigma^*_{\bx}) \geq R_s(\Sigma_{\bx})$. Since $\Sigma_{\bx}$
is assumed to be optimal, $R_s(\Sigma_{\bx}) \geq
R_s(\Sigma^*_{\bx})$, and then we know that $R_s(\Sigma^*_{\bx}) =
R_s(\Sigma_{\bx})$.
\end{proof}

Although for general $n_r$, we only can obtain the above property
for the optimization problem in Proposition \ref{Lem_No_CQI}, due
to the aforementioned difficulty below Proposition
\ref{Lem_No_CQI}. For the special MISOSE case $n_r=n_e=1$, we can
have the following secrecy capacity result as in the upcoming
Theorem. The key proof steps come as follows. First, we use
\eqref{eq_per_MIMOME_capacity} in Proposition
\ref{Pro_MIMO_per_antenna} to transform the objective function in
\eqref{eq_per_MIMOME_capacity_old} into a simpler one as in the
upcoming \eqref{eq_Cs_optimizae}. And then we can apply
Proposition \ref{Pro_Full_power} and properties of
random Gaussian vectors to obtain analytical solutions. \\
\begin{Theo}\label{Theorem_Per}\it
Subject to the per-antenna power constraints $
\{\Sigma_{\bx}\}_{ii} \leq P_i, $ $i=1,\ldots,n_t$, the secrecy
capacity optimization problem \eqref{eq_per_MIMOME_capacity} has
optimal channel input covariance matrix
\begin{equation} \label{eq_per_optimal}
\Sigma^*_{\bx}=\mathrm{diag}\left(P_1,P_2,\ldots,P_{n_t}\right)
\end{equation}
for the MISOSE wiretap channel, where the legitimate receiver and
eavesdropper has single antenna each $n_r=n_e=1$. \\
\end{Theo}

\begin{proof}
For this special case $n_r=n_e=1$, the sub-matrix of the
legitimate channel $\bH_{(n_r-n_e)}$ in
\eqref{eq_per_MIMOME_capacity} does not exist. Then the
optimization problem in \eqref{eq_per_MIMOME_capacity} can be
rewritten as
\begin{align}
&\max_{\Sigma_{\bx}}\!\!\left(\!\!\!\E_{\mathbf{g}}\!\left[\log\frac{\sigma_{\mathbf{g}}^2/\sigma_{\mathbf{h}}^2+\mathbf{g}^\dag\Sigma_{\mathbf{x}}\mathbf{g}}{\sigma_{\mathbf{g}}^2/\sigma_{\mathbf{h}}^2}\right]\!-\!\E_{\mathbf{g}}\!\left[\log(1+\mathbf{g}^\dag\Sigma_{\mathbf{x}}\mathbf{g})\right]\!\!\!\right),
\notag \\
=&\max_{\Sigma_{\bx}}\!\!\left(\!\!\!\E_{\mathbf{g}}\!\left[\log\left(1+\frac{\sigma_{\mathbf{g}}^2/\sigma_{\mathbf{h}}^2-1}{1+\mathbf{g}^\dag\Sigma_{\mathbf{x}}\mathbf{g}}\right)\right]\right)+\!\log\left(\sigma_{\mathbf{h}}^2/\sigma_{\mathbf{g}}^2\right),
\label{eq_Cs_optimizae}
\end{align}
where for clearness, we replace the notation $\bG$ in
\eqref{eq_per_MIMOME_capacity} with $\mathbf{g}^\dag\sim CN(0,
\sigma_{\bg}^2 \mathbf{I})$ to emphasize that the eavesdropper
channel is now a vector when $n_r=1$. Note that the objective
function in \eqref{eq_Cs_optimizae} is only related to the channel
vector $\bg$ for the eavesdropper, but not the channel vector for
the legitimate receiver. The above transformations make finding
analytical solutions possible.

Next, from Proposition \ref{Pro_Full_power}, without loss of
optimality, one can replace the inequality constraint
\eqref{eq_per_power_cons} by the equality constraint on the
diagonal entries of $\Sigma_{\bx}$ as
\begin{equation} \label{eq_per_equality_cons}
\left\{\Sigma_{\bx}\right\}_{ii}= P_i,
\end{equation} for
$i=1,\ldots, n_t$. Now we show that for any $\beta \leq 0$,
$\sigma > 0$, $\beta +2\sigma > 0$, subject to equality constraint
\eqref{eq_per_equality_cons}, the optimization problem
\begin{equation} \label{eq_per_MISOSE_Cs_optimizae}
\max_{\Sigma_{\bx}}\!\!\left(\!\!\!\E_{\mathbf{g}}\!\left[\log\left(1+\frac{\beta}{\sigma+\mathbf{g}^\dag\Sigma_{\mathbf{x}}\mathbf{g}}\right)\right]\right)
\end{equation}
has optimal solution \eqref{eq_per_optimal}. This fact is proved
in Appendix \ref{App_per_MISOSE_capcity}, where the key step is
cleverly applying properties of random Gaussian vectors. Note that
the objective function in \eqref{eq_Cs_optimizae} is a special
case of that in \eqref{eq_per_MISOSE_Cs_optimizae} with
$\beta=\sigma_{\mathbf{g}}^2/\sigma_{\mathbf{h}}^2-1$ and
$\sigma=1$. And then subject to \eqref{eq_per_power_cons}, the
optimal $\Sigma_{\bx}$ maximizing \eqref{eq_per_MIMOME_capacity}
(and the the original \eqref{eq_per_MIMOME_capacity_old}) with
$n_r=n_e=1$ is \eqref{eq_per_optimal}. Our claim in Theorem
\ref{Theorem_Per} is valid.
\end{proof}

Here we investigate the SNR scaling of the MISOSE channel under
the per-antenna power constrains. From Theorem \ref{Theorem_Per},
we have
\begin{coro} \label{coro_per_SNR_MISOSE}
Under the per-antenna power constraint $ \{\Sigma_{\bx}\}_{ii}
\leq P_i, $ $i=1,\ldots,n_t$, when $n_r=n_e=1$ and $\sigma_h \geq
\sigma_g > 0$, we have the following asymptotic results for the
MISOSE secrecy capacity $C^p_s$ as
\begin{equation} \label{eq_per_DoF_MISOSE}
\lim_{P^g_{max} \rightarrow
\infty}\frac{C^p_s-\log(\sigma^2_h/\sigma^2_g)}{\log P^g_{max}}=0,
\end{equation}
where $P^g_{max}=\sigma^2_g P_{max}$ with
\begin{equation} \label{eq_per_Pmax}
P_{max} \triangleq \max_{i=1,\ldots n_t}P_i
\end{equation}
being the maximum among per-antenna power constraints. \\
\end{coro}

\begin{proof}
As \eqref{eq_capacity_object}, let us denote the objective
function of optimization problem
\eqref{eq_per_MIMOME_capacity_old} as $R_s(\Sigma_{\bx})$. From
the Property (I) of Lemma \ref{lem_mono}, we know that the MISOSE
secrecy capacity under the per-antenna power constraint
$C^p_s=R_s(diag\{P_1,\ldots,P_{n_t}\}) \leq
R_s(diag\{P_{max},\ldots,P_{max}\})$. This fact is due to that
\[
diag\{P_{max},\ldots,P_{max}\}\succeq
diag\{P_1,\ldots,P_{n_t}\}\succeq \mathbf{0},
\]
because
$diag\{P_{max},\ldots,P_{max}\}-diag\{P_1,\ldots,P_{n_t}\}$ is a
diagonal matrix with non-negative entries from
\eqref{eq_per_Pmax}. By Theorem \ref{Theorem_No_CQI} we know that
$R_s(diag\{P_{max},\ldots,P_{max}\})$ is the MISOSE secrecy
capacity under total power constraint $\mathrm{Tr}(\Sigma_{\bx})
\leq n_tP_{max}$. Then by Corollary \ref{coro_total_SNR}, for
fixed $n_t$, \eqref{eq_per_DoF_MISOSE} is valid since $n_r=n_e=1$.
\end{proof}

From Corollary \ref{coro_per_SNR_MISOSE}, we know that the
multiple transmitter antennas for the MISOSE channel limitedly
help to increase the secrecy capacity under the per-antenna power
constraint. Indeed, note that \eqref{eq_per_DoF_MISOSE} is hold
even when only one transmitter antenna, which has the maximum
power constraint, is selected to transmit secret messages while
the other ones are silent. In this case, the per-antenna power
constraint corresponding to \eqref{eq_per_power_cons} are
$P_i=P_{max}$ and $P_j=0, j \neq i, \forall j$, where the $i$th
transmitter antenna is selected. When SNR is large enough, the
secrecy rate for this simple transmitter antenna selection scheme
\cite{molisch2004mimo} is within constant gap compared with the
secrecy capacity where all transmitter antennas are used.

To overcome the negative results for MISOSE channels revealed in
Corollary \ref{coro_per_SNR_MISOSE}, as in Section
\ref{Sec_total_power}, we now show that making the difference
between numbers of antennas at the legitimate receiver and
eavesdropper $n_r-n_e>0$ will help to increase the secrecy
capacity under the per-antenna power constraint. Although we
cannot find the optimal $\Sigma_{\bx}$ for
\eqref{eq_per_MIMOME_capacity_old} subject to
\eqref{eq_per_power_cons}, we use the following sub-optimal
$\Sigma_{\bx}$ to compute the secrecy capacity lower bound as
$\Sigma_{\bx}=\mathrm{diag}\left(P_1,P_2,\ldots,P_{n_t}\right)$.
This selection of $\Sigma_{\bx}$ is optimal for the MISOSE case
from Theorem \ref{Theorem_Per}, and for the MIMOME case, this
selection is also reasonable due to Proposition
\ref{Pro_Full_power}. Now we have the following result for the
secrecy capacity in Proposition \ref{Pro_MIMO_per_antenna} as
\begin{coro} \label{coro_per_SNR}
Under the per-antenna power constraint $ \{\Sigma_{\bx}\}_{ii}
\leq P_i, $ $i=1,\ldots,n_t$, when $\sigma_h \geq \sigma_g>0$, we
have the following asymptotic result for the MIMOME secrecy
capacity $C^p_s$ as
\begin{equation} \label{eq_per_DoF_MIMOME}
\lim_{P^g_{min} \rightarrow
\infty}\frac{C^p_s-n_r\log(\sigma^2_h/\sigma^2_g)}{\log
P^g_{min}} \geq \left\{\begin{array}{lc} n_r-n_e, & n_t \geq n_r > n_e\\
n_t-n_e, & n_r \geq n_t > n_e,
\end{array}\right.
\end{equation}
where $P^g_{min}=\sigma^2_g P_{min}$ with
\begin{equation}\label{eq_Pmin}
P_{min} \triangleq \min_{i=1,\ldots n_t}P_i
\end{equation}
is the minimum among per-antenna power constraints.
\end{coro}

\begin{proof}
From the definition of $P_{min}$ in \eqref{eq_Pmin}, we know that
\[
diag\{P_1,\ldots,P_{n_t}\} \succeq diag\{P_{min},\ldots,P_{min}\}
\succ \mathbf{0}.
\]
Then we know that $C^p_s \geq R_s(diag\{P_1,\ldots,P_{n_t}\}) \geq
R_s(diag\{P_{min},\ldots,P_{min}\})$, where the second inequality
comes from Property (I) of Lemma \ref{lem_mono}. By Theorem
\ref{Theorem_No_CQI} we know that
$R_s(diag\{P_{min},\ldots,P_{min}\})$ is the MIMOME secrecy
capacity under total power constraint $\mathrm{Tr}(\Sigma_{\bx})
\leq n_tP_{min}$. Then by Corollary \ref{coro_total_SNR}, for
fixed $n_t$, \eqref{eq_per_DoF_MIMOME} is valid.
\end{proof}

From \eqref{eq_per_DoF_MIMOME}, as discussions under Corollary
\ref{coro_total_SNR}, adding enough number of legitimate-receiver
antennas $n_r$ (also enough number of transmit antennas $n_t$) is
very important for increasing the secrecy capacity $C^p_s$.
Moreover, \eqref{eq_per_DoF_MIMOME} also reveals some antenna
selection rules for MIMOME wiretap channel under the per-antenna
power constraint. For example, when $n_t
> n_r
> n_e$, the transmitter can select the largest
$n_r$ antennas to transmit the secret message. Each selected
antenna transmits with all its allowable power. When the allowable
power of all selected antenna is high enough, this simple transmit
antenna selection scheme can achieve good secrecy rate performance
according to \eqref{eq_per_DoF_MIMOME}.

Finally, we develop algorithms to solve the MIMOME secrecy
capacity optimization problems \eqref{eq_per_MIMOME_capacity_old}
subject to the per-antenna power constraint
\eqref{eq_per_power_cons}. To do this, as in Proposition
\ref{Pro_MIMO_per_antenna}, by using the same marginal channel
matrix $\bH'$ in \eqref{eq_H_divide}, we first transfer the
stochastic optimization problem \eqref{eq_per_MIMOME_capacity_old}
into a concave one \eqref{eq_per_MIMOME_capacity}. Next, from
Proposition \ref{Pro_Full_power}, one can set the inequalities in
constraints \eqref{eq_per_power_cons} with equalities as
\eqref{eq_per_equality} to further simplify the problem. Then we
can use algorithms similar to those in \cite{huh2010multiuser} to
find the optimal values for \eqref{eq_per_MIMOME_capacity_old} by
matrix calculus
\cite{hjorungnes2011complex}\cite{payaro2009hessian}. The gap to
the optimal value is within $n_t/t$, where $n_t$ is the number of
transmit antenna and $t>0$ is a parameter to prevent the algorithm
approaching non-semidefine $\Sigma_{\bx}$. Note that our problem
can not be simplified by the
multiple-access-channels-broadcast-channels duality as
\cite{huh2010multiuser}. Though wiretap channels are similar to
the broadcast channels, there may be no corresponding dualities
for the wiretap channels.

\begin{figure*}[t]
\begin{align}
\nabla_x
f_t=\E_{\mathbf{H}'}\bigg[\big((\textbf{H}^{'})^\dagger\otimes(\textbf{H}^{'})^T\big)vec&\bigg(\big((\textbf{I}+\textbf{H}^{'}\Sigma_{\bx}(\textbf{H}^{'})^\dagger)^{-1}\big)^T\bigg)-\big(\bG^\dagger\otimes\bG^T\big)vec\bigg(\big((\textbf{I}+\bG\Sigma_{\bx}\bG^\dagger)^{-1}\big)^T\bigg)\bigg]+\frac{1}{t}vec\left(\left({\Sigma_{\bx}}^{-1}\right)^T\right)\label{eq_graft}
\end{align}

\begin{align}
\nabla^2_{xx}f_t
&=\E_{\mathbf{H}'}~\bigg[-((\textbf{H}^{'})^\dagger\otimes(\textbf{H}^{'})^T)\bigg(\big(\textbf{I}+\textbf{H}^{'}\Sigma_{\bx}(\textbf{H}^{'})^\dagger)\otimes\big((\textbf{I}+\textbf{H}^{'}\Sigma_{\bx}(\textbf{H}^{'})^\dagger)^{T}\bigg)^{-1}\cdot
\mathbf{K}_p
\bigg(((\textbf{H}^{'})^\dagger)^T\otimes\textbf{H}^{'}\bigg) \nonumber \\
&~~~~~~~~~~~+(\bG^\dagger\otimes\bG^T)\bigg(\big(\textbf{I}+\bG\Sigma_{\bx}\bG^\dagger)\otimes\big(\textbf{I}+\bG\Sigma_{\bx}\bG^\dagger)^{T}\bigg)^{-1}\cdot
\mathbf{K}_p
\bigg((\bG^\dagger)^T\otimes\bG\bigg)\bigg]-\frac{1}{t}(\Sigma_{\bx}\otimes\Sigma_{\bx}^T)^{-1}\mathbf{K}_p
\label{eq_gragraft}
\end{align}

\hrulefill
\end{figure*}

Here comes the reformulation of the secrecy capacity optimization
problem \eqref{eq_per_MIMOME_capacity_old} subject to
\eqref{eq_per_power_cons} for our algorithm. First, we use the
following objective function as
\begin{equation} \label{eq_barrieried_objective}
f_t(\Sigma_{\bx})=\tilde{R}_s(\Sigma_{\bx})+\frac{1}{t}\log|\Sigma_{\bx}|,
\end{equation}
where $\tilde{R}_s(.)$ is the concave objective function in
\eqref{eq_per_MIMOME_capacity} and $\frac{1}{t}\log|\Sigma_{\bx}|$
is the logarithmic barrier. As for constraint
\eqref{eq_per_equality} (per-antenna power constraint
\eqref{eq_per_power_cons} with equality), we also rewrite it as
the following equality constraint
\begin{equation} \label{eq_Alg_constraint}
\mathbf{A}\cdot \emph{vec}(\Sigma_{\bx}) =
[~P_1~\ldots~P_{n_t}]^T,
\end{equation}
where the $n_t \times (n_t)^2$ matrix $\mathbf{A}$ has entries
being $1$ for those corresponds to $\{\Sigma_{\bx}\}_{ii}, i=1
\ldots n_t$, and entries being $0$ else, and the vectorization
operator for a matrix $\emph{vec}(.)$ is defined as
\cite{hjorungnes2011complex}\cite{payaro2009hessian}. For example,
when $n_t=2$
\[\mathbf{A}=
\left[\begin{array}{cccc}
1 & 0 & 0 & 0\\
0 & 0 & 0 & 1
\end{array}\right].
\]

Now we focus on the detailed steps for the optimization of
$f_t(\Sigma_{\bx})$ in \eqref{eq_barrieried_objective} subject to
equality constraint \eqref{eq_Alg_constraint}. In the following,
the $\nu$ is the $n_t \times 1$ Lagrange multiplier associated to
equality constraint \eqref{eq_Alg_constraint}.

\begin{itemize}

\item[1)] Initialize $\Sigma^{(0)}_{\bx}$ and $\nu^{(0)}$

\item[2)] By invoking the matrix calculus
\cite{hjorungnes2011complex}\cite{payaro2009hessian}, compute the
residue as
\[
\br= \left[\begin{array}{c}
\nabla_xf_t+{\textbf{A}}^T\nu\\
\textbf{A}\cdot \emph{vec}(\Sigma_{\bx})-[~P_1~\ldots~P_{n_t}]^T
\end{array}\right]
\]
where $\nabla_x$ means the gradient of function respect to vector
$\emph{vec}(\Sigma_{\bx})$, and $\nabla_xf_t$ is given in
\eqref{eq_graft} with $\bH'$ from \eqref{eq_H_divide} and
$\otimes$ being the Kronecker product. The direction
$[\Delta_{x},\Delta_{\nu}]^T$ to update
$[\emph{vec}(\Sigma_{\bx}),\nu]^T$ is computed by the KKT matrix
\cite{huh2010multiuser} as
\[
\left[\begin{array}{c}
\Delta_{x}\\
\Delta_\nu
\end{array}\right]
=- \left[\begin{array}{cc}
\nabla_{xx}^2f_t & \textbf{A}^T\\
\textbf{A} & 0
\end{array}\right]^{-1} \br,
\]
where $\nabla_{xx}^2f_t$ is explicitly given in
\eqref{eq_gragraft}, with $\mathbf{K}_p$ being the permutation
matrix such that $vec(\mathbf{A}^T ) = \mathbf{K}_p
vec(\mathbf{A})$ as defined in \cite[(51)]{payaro2009hessian}.

\item[3)] For the $n+1$ step, we compute
$\left(\emph{vec}\left(\Sigma_{\bx}^{(n+1)}\right),\nu^{(n+1)}\right)
=
\left(\emph{vec}\left(\Sigma_{\bx}^{(n)}\right)+s\Delta_x,\nu^{(n)}+s\Delta\nu\right)$,
where $s$ is the step size found by the backtracking line search.

\item[4)] If the norm of residue
$\|\br\|_{_2}<\varepsilon$,~increase $t$ by a factor $\gamma=1.5$
and go to step 5. If not,~go back to step 2.

\item[5)] Stop when the gap $\frac{n_t}{t}<\varepsilon$,~if
not,~go back to step 2.
\end{itemize}

Note that steps 2 to 4 are the infeasible start Newton steps, and
the algorithm will converge since objective function $f_t(.)$ in
\eqref{eq_barrieried_objective} is concave (both $\tilde{R}_s(.)$
from \eqref{eq_per_MIMOME_capacity} and the logarithmic barrier
are concave) \cite{huh2010multiuser}. And following
\cite{huh2010multiuser}, one can show that the gap to the optimal
value of \eqref{eq_per_MIMOME_capacity_old} is $\frac{n_t}{t}$ as
in step 5, because the dimension of $\Sigma_{\bx}$ in the
logarithmic barrier $\frac{1}{t}\log|\Sigma_{\bx}|$ is $n_t \times
n_t$.

Compared with the algorithm in \cite{li2010optimal} which also
considers the per-antenna power constraint, our algorithm can
guarantee that the gap to the optimal value is vanishing when $t
\rightarrow \infty$ while that in \cite{li2010optimal} can not
guarantee such an optimality. Moreover, \cite{li2010optimal}
requires full CSIT of both legitimate channel and eavesdropper
while our statistical CSIT requirement is more practical.

\section{Simulations}\label{Sec_simulation}
In this section, we provide numerical and simulation results for
our theoretical claims in previous sections. In all figures
presented, noise at the legitimate receiver and eavesdropper has
unit variance each. The SNR is then defined as the total
transmitted power over all antennas in the dB scale. The ratio of
the legitimate channel variance over the eavesdropper channel
variance in \eqref{eq_Rayleigh} is $\sigma^2_h/\sigma^2_g=4$. All
channel matrices are Rayleigh faded.

First, in Fig. \ref{fig_capacity_different_P}-\ref{fig_alpha}, we
show the numerical results for secrecy capacities $C^{\;t}_s$
under the total power constraints \eqref{power_cons1} in Theorem
\ref{Theorem_No_CQI}. In Fig. \ref{fig_capacity_different_P}, we
compare the MIMOME secrecy capacities $C^{\;t}_s$ under different
combinations of number of antennas. Three different combinations
$(n_t,n_r,n_e)$ with fixed $n_t=4$ and $n_t \geq n_r$ are
considered, where $n_t,n_r$ and $n_e$ respectively are the number
of antennas at the transmitter, the legitimate receiver, and the
eavesdropper. Consistent with \cite{Secrecy_sc}, when $n_r=n_e=1$,
the MISOSE secrecy capacities do not scale with SNR and converges
at high SNR. Same phenomenon also happens even when $n_r=n_e=2$.
Thus it may be a waste of resource by increasing the SNR for the
MIMOME channel with $n_r=n_e$. To overcome this drawback, one can
increase $n_r$ to make it larger than $n_e$. These observations
meet our results in Corollary \ref{coro_total_SNR}. Also with
fixed number of transmit antennas $n_t>n_e$, increasing the
difference $n_r-n_e$ is very helpful to increase the secrecy
capacities in all SNR regimes. Similarly, as in Fig.
\ref{fig_capacity_different_P}, in Fig.
\ref{fig_capacity_different_n_t} we compare the MIMOME secrecy
capacities $C^{\;t}_s$ in Theorem \ref{Theorem_No_CQI} but with
fixed $n_r=5$ and $n_r \geq n_t$. The MISOSE secrecy capacities
with $(n_t,n_r,n_e)=(5,1,1)$ are also depicted in Fig
\ref{fig_capacity_different_n_t}, and again they do not scale with
the SNR. The SNR scaling in MIMOME channels can be obtained when
$n_r
> n_e$ but not $n_r=n_e$. And with fixed $n_r>n_e$ (but not $n_r=n_e$), increasing the
difference $n_t-n_e$ is very helpful to increase the secrecy
capacities even at medium SNR regimes.

Next, we consider the MIMOME secrecy capacities $C^{\;t}_s$ in
Theorem \ref{Theorem_No_CQI} with single transmit antenna in Fig.
\ref{fig_capacity_different_n_e}. As predicted in Corollary
\ref{coro_total_SNR}, when $n_t=1$ there will be no SNR scaling no
matter $n_r>n_e$ or $n_r=n_e$. However, one can observe that
channels with larger $n_r-n_e$ (but not larger $n_r$) will have
higher secrecy capacities $C^{\;t}_s$. Next, in Fig.
\ref{fig_alpha}, we show that all available transmit power $P$
should be used for MIMOME channels with total power constraints
\eqref{power_cons1}. We plot the $R_s(\alpha \mathbf{I})$s in
\eqref{eq_capacity_object}, the secrecy rates with channel input
matrices $\Sigma_{\bx}=\alpha \bI$. The maximum allowable transmit
power $P=20$ dB and $0 \leq \alpha \leq P/n_t$. From
Fig.\ref{fig_alpha}, the secrecy rate $R_s(\alpha \mathbf{I})$
increases monotonically with $\alpha$. Thus when $\alpha$ equals
to its maximum $P/n_t$, the secrecy rate $R_s(P/n_t \mathbf{I})$
is the secrecy capacity, as in the proof of Theorem
\ref{Theorem_No_CQI}.

\begin{figure} [t]
\centering
\epsfig{file=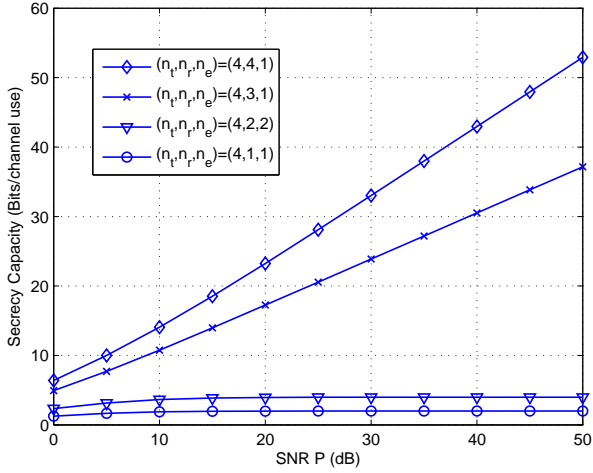,width=0.5\textwidth}
\caption{Under total power constraints over all transmit antennas
\eqref{power_cons1}, the secrecy capacities versus SNRs with
different number of antennas $(n_t \geq n_r)$.}
\label{fig_capacity_different_P}
\end{figure}
\begin{figure} [t]
\centering \epsfig{file=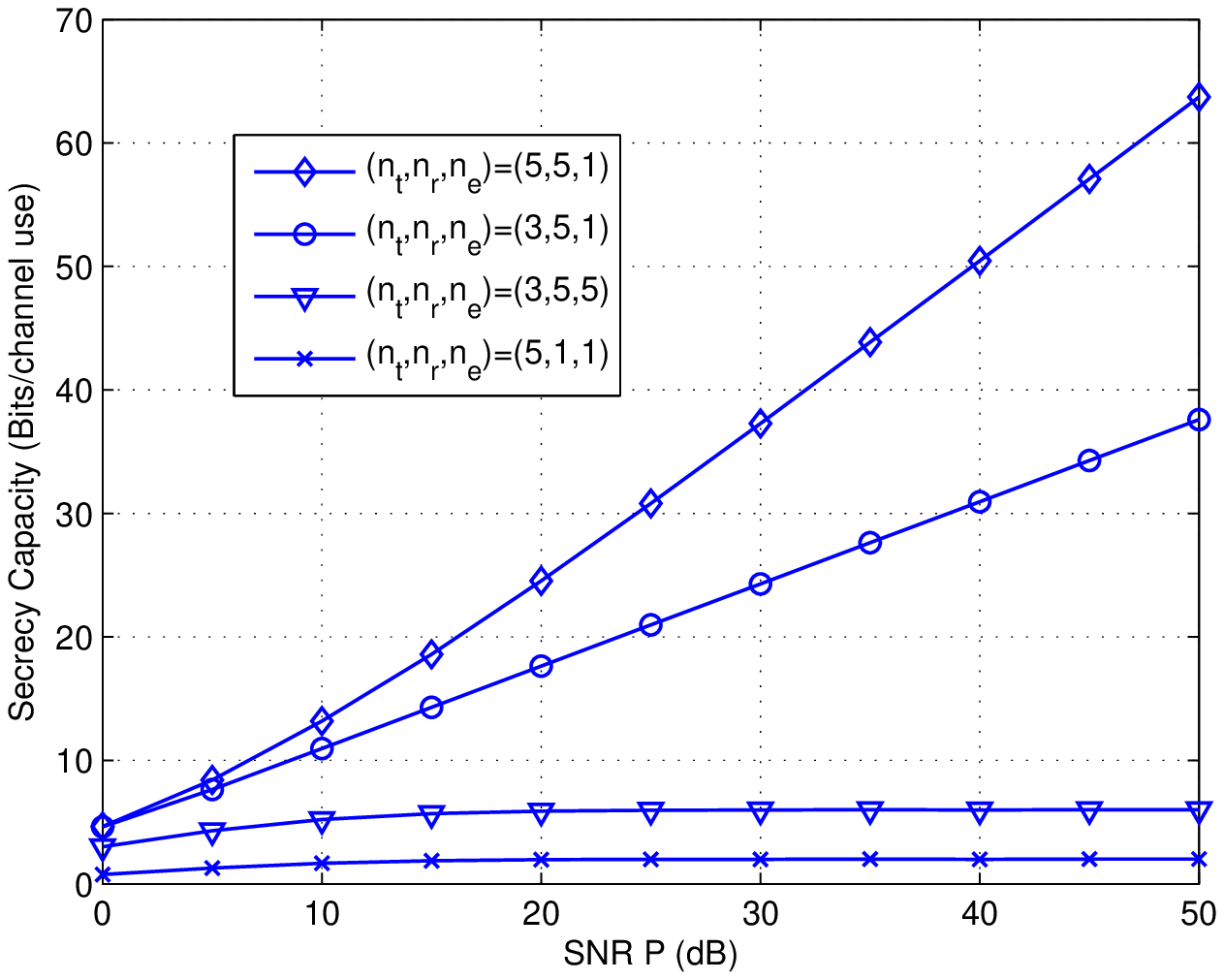,
width=0.5\textwidth} \caption{Under total power constraints over
all transmit antennas \eqref{power_cons1}, the secrecy capacities
versus SNRs with different number of antennas $(n_r \geq n_t)$.}
\label{fig_capacity_different_n_t}
\end{figure}
\begin{figure} [t]
\centering \epsfig{file=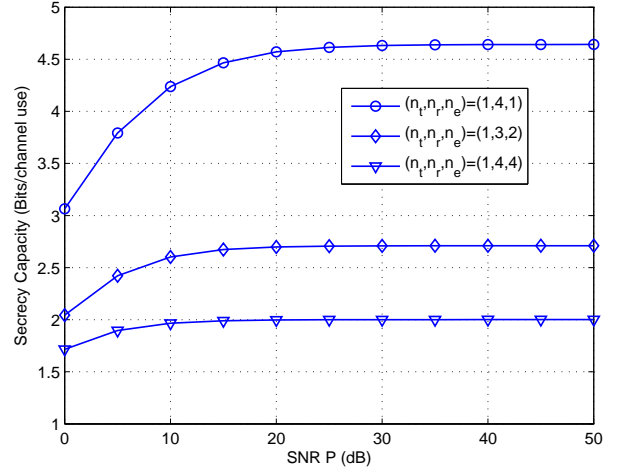,
width=0.5\textwidth} \caption{Under total power constraints over
all transmit antennas \eqref{power_cons1}, the secrecy capacities
versus SNRs with single transmit antenna $(n_t=1)$.}
\label{fig_capacity_different_n_e}
\end{figure}

\begin{figure} [t]
\centering \epsfig{file=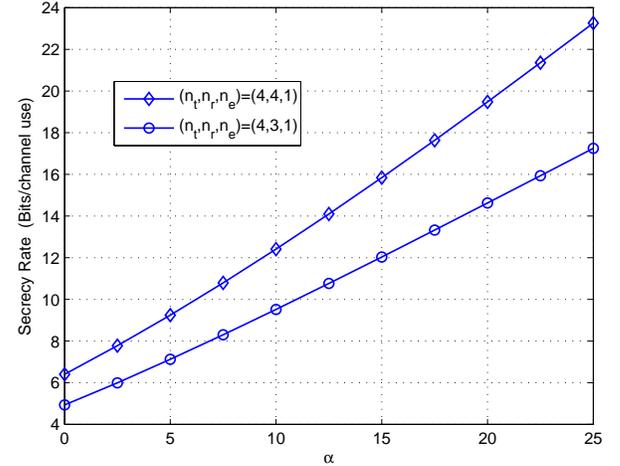, width=0.5\textwidth}
\caption{Under total power constraints over all transmit antennas
\eqref{power_cons1}, the MIMOME secrecy rates $R_s(\alpha
\mathbf{I})$ in \eqref{eq_capacity_object} (with channel input
matrix $\Sigma_{\bx}=\alpha \bI$) versus $\alpha$ where $0 \leq
\alpha \leq P/n_t$.} \label{fig_alpha}
\end{figure}

In Fig. \ref{fig_MIMOME_vs_MISOSE} and
\ref{fig_linesearch_interation_times}, we studies the wiretap
channels under the per-antenna power constraints
\eqref{eq_per_power_cons}. In Fig.\ref{fig_MIMOME_vs_MISOSE}, the
secrecy rates of wiretap channels with $(n_t,n_r,n_e)=(2,2,1)$
(the objection function of \eqref{eq_per_MIMOME_capacity_old} with
channel input covariance matrices $\Sigma_{\bx}=diag(P_1,P_2)$)
are compared with the secrecy capacities $C^{p}_s$ of wiretap
channels with $(n_t,n_r,n_e)=(2,1,1)$ (characterized in Theorem
\ref{Theorem_Per}). As predicted in Corollary
\ref{coro_per_SNR_MISOSE}, the MISOSE channels under constraints
\eqref{eq_per_power_cons} cannot have SNR scaling. However, as
stated in Corollary \ref{coro_per_SNR}, the SNR scaling can be
obtained for the wiretap channel with $(n_t,n_r,n_e)=(2,2,1)$ even
we use the suboptimal $\Sigma_{\bx}$ to compute the secrecy rate
in Fig.\ref{fig_MIMOME_vs_MISOSE}. In
Fig.\ref{fig_linesearch_interation_times}, we use the algorithm in
Section \ref{Sec_per_antenna} to numerically solve the secrecy
capacity optimization problem \eqref{eq_per_MIMOME_capacity_old}
subject to constraint \eqref{eq_per_power_cons}, where
$(n_t,n_r,n_e)=(2,2,1)$ and $\epsilon=10^{-4}$. The SNR is
$P_1+P_2$ in dB scale. The proposed algorithms converge in few
iterations. Note that the final converged values are very close
the secrecy rates computed using channel input covariance matrices
$\Sigma_{\bx}=(P_1,P_2)$. This fact implies that the secrecy rates
for $(n_t,n_r,n_e)=(2,2,1)$ plotted in Fig.
\ref{fig_MIMOME_vs_MISOSE} are very close to the secrecy
capacities under per-antenna power constraints
\eqref{eq_per_power_cons}.

Finally, in Fig. \ref{fig_total_P_vs_per_antenna_P}, we compare
$C^{\;t}_s$ and $C^p_s$, the MISOSE secrecy capacities under the
total power constraints \eqref{power_cons1} and per-antenna power
constraints \eqref{eq_per_power_cons} in Theorem
\ref{Theorem_No_CQI} and \ref{Theorem_Per}, respectively. Since
under total power constraints \eqref{power_cons1}, the
transmitters are free to allocate power between transmit antennas,
the corresponding secrecy capacities are higher than those under
per-antennas power constraints \eqref{eq_per_power_cons}
($C^{\;t}_s \geq C^p_s$ when $P=P_1+P_2$). We also observe that
when the ratio $P_1/P_2$ are larger, the secrecy capacities
$C^p_s$ become smaller when $(n_t, n_r, n_e)=(2,1,1)$. This is due
to that the multiple transmit antennas act more like a single
antenna when $P_1/P_2$ is larger. Similar observations can be
obtained from the secrecy rates under power constraints
\eqref{eq_per_power_cons} and $(n_t, n_r, n_e)=(2,2,1)$ in Fig.
\ref{fig_MIMOME_vs_MISOSE}.

\begin{figure} [t]
\centering \epsfig{file=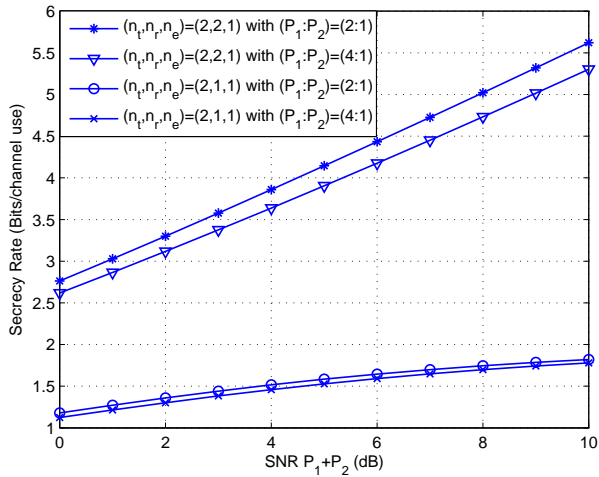,
width=0.5\textwidth} \caption{Under per-antenna power constraints
\eqref{eq_per_power_cons}, secrecy rates of wiretap channels with
$(n_t,n_r,n_e)=(2,2,1)$ versus secrecy capacities of wiretap
channels with $(n_t,n_r,n_e)=(2,1,1).$}
\label{fig_MIMOME_vs_MISOSE}
\end{figure}

\begin{figure} [t]
\centering \epsfig{file=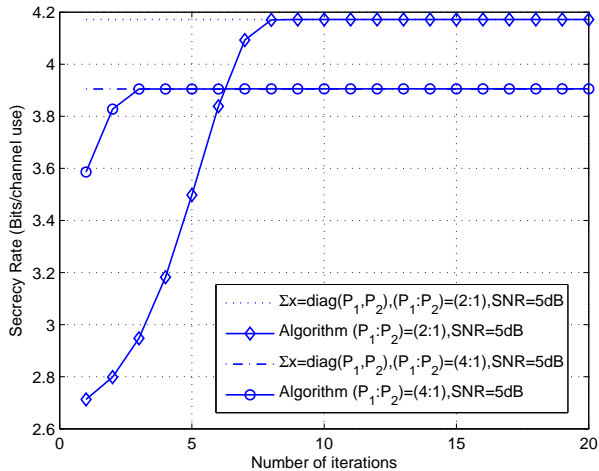,
width=0.5\textwidth} \caption{Convergence of the proposed
iterative algorithms for computing the secrecy capacities under
constraints \eqref{eq_per_power_cons}, where
$(n_t,n_r,n_e)=(2,2,1)$.} \label{fig_linesearch_interation_times}
\end{figure}

\begin{figure} [t]
\centering \epsfig{file=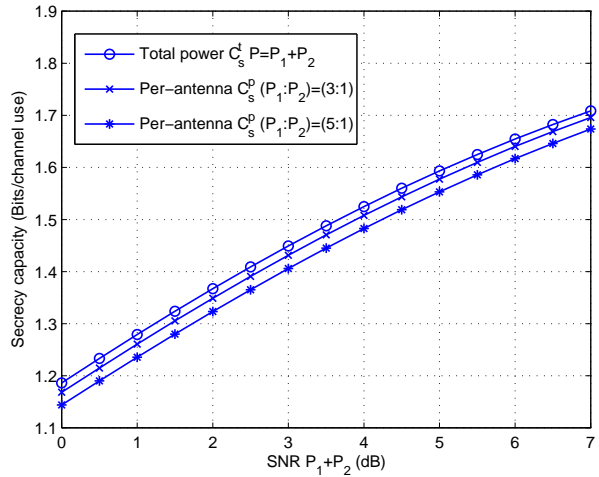,
width=0.5\textwidth} \caption{MISOSE secrecy capacities under the
total power constraints \eqref{power_cons1} and per-antennas power
constraints \eqref{eq_per_power_cons}, where
$(n_t,n_r,n_e)=(2,1,1)$.} \label{fig_total_P_vs_per_antenna_P}
\end{figure}

\section{Conclusion}\label{Sec_conclusion}
In this paper, under two different power constraints, the secrecy
capacities in fast fading Rayleigh MIMOME wiretap channels with
only the statistics of CSIT of both the legitimate and
eavesdropper channels were considered. When antennas of the
legitimate receiver were more than (or equal to) those of the
eavesdropper, under the total power constraint, we fully
characterized the MIMOME secrecy capacity. Under the per-antenna
power constraint, we also showed the secrecy capacity for the
MISOSE channel. These results are the first secrecy capacity
results for multiple antenna wiretap channels with partial CSIT.

\appendix

\subsection{Proof of Lemma \ref{lem_mono}} \label{app_mono}
We first focus on Property (I). First, let us consider the case
where $\Sigma^2_x \succeq \Sigma^1_x \succ \mathbf{0}$. We begin
our proof by transforming the objective function in
\eqref{eq_MISOSE_capacity} into the upcoming
\eqref{eq_total_alpha3}. Since $\bH'$ in \eqref{eq_H_divide} has
the same distribution as $\bH$, we rewrite the objective function
in \eqref{eq_MISOSE_capacity} as
\[
\E_{\bH}\!\left[\!\log\left|\bI+\bH'\Sigma_{\bx}(\bH')^\dag\!\right|\!\right]\!-\!\E_{\bG}\!\left[\!\log\!\left|\bI+\bG\Sigma_{\bx}\bG^\dag\!\right|\!\right].
\]
And by the matrix equality $|\bI+\bA\bB|=|\bI+\bB\bA|$, the above
equation equals to
\begin{equation} \label{eq_total_alpha1}
\E_{\bH'}\!\left[\!\log\left|\bI+\Sigma_{\bx}(\bH')^\dag\bH'\!\right|\!\right]\!-\!\E_{\bG}\!\left[\!\log\!\left|\bI+\Sigma_{\bx}\bG^\dag\bG\!\right|\!\right].
\end{equation}
For any $\Sigma_{\bx} \succ \mathbf{0}$, from \eqref{eq_H_divide},
\begin{align}
&\E_{\bH'}\!\left[\!\log\left|\bI+\Sigma_{\bx}(\bH')^\dag\bH'\!\right|\!\right]
\notag \\
=&\E_{\bH'}\!\left[\!\log\left|\bI+\Sigma_{\bx}
\left(\bH^\dag_{(n_r-n_e)}\bH_{(n_r-n_e)}+\frac{\sigma^2_h}{\sigma^2_g}\bG^\dag\bG\right)\!\right|\!\right].
\notag
\end{align}
Substituting the above equalities into \eqref{eq_total_alpha1}, it
then equals to
\begin{equation} \label{eq_total_alpha2}
\E_{\bH'}\!\left[\!\log\left|\bI+\left(\bI+\Sigma_{\bx}\bG^\dag\bG\!\right)^{-1}\Sigma_{\bx}\left((\tilde{\bH}')^\dag\tilde{\bH}'\right)\!\right|\!\right],
\end{equation}
where
\begin{equation} \label{eq_total_alpha_H}
(\tilde{\bH}')^\dag\tilde{\bH}'=\bH^\dag_{(n_r-n_e)}\bH_{(n_r-n_e)}+\left(\frac{\sigma^2_h}{\sigma^2_g}-1\right)\bG^\dag\bG.
\end{equation}
Note that since since $\sigma^2_h \geq \sigma^2_g$, the RHS of
\eqref{eq_total_alpha_H} is positive semi-definite. And thus we
can always find $\tilde{\bH}'$ to make \eqref{eq_total_alpha_H}
valid \cite{Horn_matrix_analysis}. Then we can rewrite
\eqref{eq_total_alpha2} as
\begin{equation} \label{eq_total_alpha3}
\E_{\bH'}\!\left[\!\log\left|\bI+\tilde{\bH}'\left(\Sigma_{\bx}^{-1}+\bG^\dag\bG\!\right)^{-1}(\tilde{\bH}')^\dag\!\right|\!\right].
\end{equation}

Now for any $\Sigma^2_x \succeq \Sigma^1_x \succ \mathbf{0}$, we
know that $ \left((\Sigma^2_x)^{-1}+\bG^\dag\bG\!\right)^{-1}
\succeq \left((\Sigma^1_x)^{-1}+\bG^\dag\bG\!\right)^{-1} \succ
\mathbf{0}, $ and thus
\begin{align}
&\tilde{\bH}'\left((\Sigma^2_x)^{-1}+\bG^\dag\bG\!\right)^{-1}(\tilde{\bH}')^\dag
\notag \\
\succeq &
\tilde{\bH}'\left((\Sigma^1_x)^{-1}+\bG^\dag\bG\!\right)^{-1}(\tilde{\bH}')^\dag
\overset{(a)}\succ \mathbf{0}. \label{eq_total_alpha4}
\end{align}
Note that the above relationship is valid for every realization of
the random channels, and from \cite{Horn_matrix_analysis}, we know
that
\begin{align} \notag
&\E_{\bH'}\!\left[\!\log\left|\bI+\tilde{\bH}'\left((\Sigma^2_x)^{-1}+\bG^\dag\bG\!\right)^{-1}(\tilde{\bH}')^\dag\!\right|\!\right]
\notag \\ \geq
&\E_{\bH'}\!\left[\!\log\left|\bI+\tilde{\bH}'\left((\Sigma^1_x)^{-1}+\bG^\dag\bG\!\right)^{-1}(\tilde{\bH}')^\dag\!\right|\!\right].
\notag
\end{align}
Thus \eqref{eq_mono_sigma_x} is valid if $\Sigma^2_x \succeq
\Sigma^1_x \succ \mathbf{0}$.

Now we only need to further prove that if $\Sigma^2_x \succeq
\Sigma^1_x \succeq \mathbf{0}$, when $\Sigma^1_x$ is singular,
$R_s(\Sigma^2_x) \geq R_s(\Sigma^1_x)$. Then our claim in Property
(I) is valid. Here we prove the case where both $\Sigma^1_x$ and
$\Sigma^2_x$ are singular, while the case where only $\Sigma^1_x$
is singular can be proved similarly. First, $\forall \alpha>0$, as
shown in the previous paragraph,
$R_s(\Sigma^2_x+\alpha\bI)-R_s(\Sigma^1_x+\alpha\bI) \geq 0$ when
$\Sigma^2_x \succeq \Sigma^1_x \succeq \mathbf{0}$ since
$\Sigma^2_x+\alpha\bI \succeq \Sigma^1_x+\alpha\bI \succ
\mathbf{0}$. Next, following methods in
\cite[P.3957]{NIT_Weingarten_CRegion_sIT04}, we have $\lim_{\alpha
\rightarrow 0} R_s(\Sigma^1_x+\alpha\bI)=R_s(\Sigma^1_x)$ with
$\alpha>0$. This fact comes from the continuity of the log det
functions in \eqref{eq_capacity_object} over positive definite
matrices \cite{NIT_Weingarten_CRegion_sIT04}. Then we know that
$R_s(\Sigma^2_x)-R_s(\Sigma^1_x)=\lim_{\alpha \rightarrow 0}
\left(R_s(\Sigma^2_x+\alpha\bI)-R_s(\Sigma^1_x+\alpha\bI)\right)
\geq 0$. This concludes our proof for Property (I). As for
Property (II), it can be easy obtained from (\ref{eq_total_alpha4}
a)

\subsection{Proof of Colloray \ref{coro_total_SNR}}
\label{App_coro_total_SNR} We only prove the case where $n_t \geq
n_r > n_e$, since the rest two cases can be proved similarly. From
Theorem \ref{Theorem_No_CQI}, we know that the MIMOME secrecy
capacity $C^{\;t}_s$ equals to
\[
C^{\;t}_s=\E_{\bH}\!\left[\!\log\left|\bI+\frac{P}{n_t}\bH\bH^\dag\!\right|\!\right]\!-\!\E_{\bG}\!\left[\!\log\!\left|\bI+\frac{P}{n_t}\bG\bG^\dag\!\right|\!\right]
\]
From \cite{Book_TulinoVerdu}, one can transform the previous
equation as
\[
C^{\;t}_s=n_r
\E\left[\log\left(1+\frac{\sigma^2_hP}{n_t}\lambda\right)\right]-n_e
\E\left[\log\left(1+\frac{\sigma^2_gP}{n_t}\lambda\right)\right],
\]
where $\lambda$ is an unordered eignvalue of a complex Wishart
matrix with $n_t$ degrees of freedom and covariance matrix $\bI$.
With the definition $P_g=\sigma^2_gP$,
\[
C^{\;t}_s=n_r
\E\left[\log\left(1+\frac{\sigma^2_h}{\sigma^2_g}\frac{P_g}{n_t}\lambda\right)\right]-n_e
\E\left[\log\left(1+\frac{P_g}{n_t}\lambda\right)\right].
\]
Note that $\sigma^2_h/\sigma^2_g \geq 1$, then
\[
\lim_{P_g \rightarrow
\infty}\frac{C^{\;t}_s-n_r\log(\sigma^2_h/\sigma^2_g)}{\log
P_g}=n_r-n_e,\] for fixed $n_t$.

\subsection{Proof of Proposition \ref{Pro_MIMO_per_antenna}}
\label{app_Pro_MIMO_per_antenna} Following
\cite{NIT_Weingarten_CRegion_sIT04}, we can show that the result
in Lemma \ref{Lem_No_CQI} is still valid under
\eqref{eq_per_power_cons}, with proof based on replacing
\eqref{eq_per_power_cons} with the covariance matrix constraint
$\E(\bx\bx^\dag) \preceq \bS$. To be more specific,
$C^p_s(P_1,\ldots,P_{n_t})=max_{\bS \in \mathds{S}} C_s(\bS)$,
where $C^p_s(P_1,\ldots,P_{n_t})$ is the secrecy capacity under
per-antenna power constraint \eqref{eq_per_power_cons}, $C_s(\bS)$
is the secrecy capacity under the covariance matrix constraint
$\E(\bx\bx^\dag) \preceq \bS$, and the set
$\mathds{S}=\{\bS|\{\bS\}_{ii} \leq P_i, i=1,\ldots, n_t, \bS
\succeq \b0 \}$. Next, we can show that Gaussian signal $\bx$
without prefixing $U \equiv \bx$ is still secrecy capacity
achieving with respect to $C_s(\bS)$. The proof is the same as
that of Lemma \ref{Lem_No_CQI}, expect for proving that under
covariance matrix constraint instead of \eqref{power_cons1},
Gaussian $\bx$ maximizes \eqref{eq_MIMOME_capacity_I}. To prove
this fact, from (\ref{eq_Convex} b), we know that the maximization
in \eqref{eq_MIMOME_capacity_I} equals to $\max_{\bx}
h(\frac{\sigma_{\bg}}{\sigma_{\bh}}\by'|\bz,\bH')$. Given a
realization of $\bH'=\underline{\bH}'$, we can show that Gaussian
$\bx$ will make $(\by',\bz)$ jointly Gaussian and maximize
$h(\frac{\sigma_{\bg}}{\sigma_{\bh}}\by'|\bz,\bH'=\underline{\bH}')$
for all $\bx$ satisfying $\E(\bx\bx^\dag) = \bS'$ and $\bS'
\preceq \bS$, by modifying the proof of \cite[Lemma
1]{Khisti_MIMOME} with \cite[Theorem 8.6.5]{Book_Cover}. Moreover,
since $\bx$ is independent of $\bH'$, we know that Gaussian $\bx$
will maximize $h(\frac{\sigma_{\bg}}{\sigma_{\bh}}\by'|\bz,\bH')$
and thus maximize \eqref{eq_MIMOME_capacity_I} subject to the
$\E(\bx\bx^\dag) \preceq \bS$ (thus also for the per-antenna power
constraint \eqref{eq_per_power_cons}).

Substitute the optimal Gaussian $\bx$ with covariance matrix
$\Sigma_{\bx}$ into $I(\bx;\by|\bH)-I(\bx;\by|\bG)$, we have the
secrecy capacity formulae as \eqref{eq_per_MIMOME_capacity_old}
subject to \eqref{eq_per_power_cons}. Next, by replacing channel
matrix $\bH$ in \eqref{eq_per_MIMOME_capacity_old} with the same
distribution one $\bH'$ in \eqref{eq_H_divide}, we have
\eqref{eq_per_MIMOME_capacity} according to the proof in Appendix
\ref{app_mono} (steps for reaching \eqref{eq_total_alpha2}).
Finally, following the first paragraph of Theorem
\ref{Theorem_No_CQI}'s proof, we know that the secrecy capacity
\eqref{eq_per_MIMOME_capacity} under \eqref{eq_per_power_cons} is
concave in $\Sigma_{\bx}$.

\subsection{Solution of the optimization problem \eqref{eq_per_MISOSE_Cs_optimizae} subject to the equality constraint
\eqref{eq_per_equality_cons}} \label{App_per_MISOSE_capcity}

In this appendix, we show that subject to constraint
\eqref{eq_per_equality_cons}, the optimization problem
\eqref{eq_per_MISOSE_Cs_optimizae} has optimal solution
\eqref{eq_per_optimal}. We only consider the case with $\beta<0$
because $\beta=0$ is a trivial case. The following proof is given
by mathematical induction, which shows that the off-diagonal terms
of $\Sigma_{\bx}$ subject to \eqref{eq_per_equality_cons} should
be zeros.

(i) Let us first consider the case with $n_t=2$. By defining
$f(x)\triangleq \log(1+x)$ to simplify the notations, the
objective function of \eqref{eq_per_MISOSE_Cs_optimizae} can be
written as
\begin{align} \label{eq_in_lemma1}
\E_\bg\!\!\left[\!\!f\!\!\left(\!
\frac{\beta}{\sigma\!+\!|g_1|^2\{\!\Sigma_{\bx}\!\}_{\!11}\!+\!|g_2|^2\{\!\Sigma_{\bx}\!\}_{\!22}\!+\!g_{1}^*\{\!\Sigma_{\bx}\!\}_{\!12}g_{2}\!+\!g_{2}^*\{\!\Sigma_{\bx}\!\}_{\!21}g_{1}}
\!\right)\!\!\right].
\end{align}
Notice that, since $g_1$ and $g_2$ are independent zero-mean
Gaussian, the expectation of \eqref{eq_in_lemma1} over $\bg=[g_1
\; g_2]^T$ would not change by replacing $g_1$ with $-g_1$.
Therefore, with $\{\Sigma_\bx\}_{11}= P_1$ and
$\{\Sigma_\bx\}_{22}= P_2$, the expectation in
\eqref{eq_in_lemma1} can be written as
\begin{align*}
&\E_{\bg}\left[f\left(
\frac{\beta}{\sigma\!+\!|g_1|^2P_1\!+\!|g_2|^2P_2\!+\!g_{1}^*\{\Sigma_\bx\}_{12}g_{2}\!+\!g_{2}^*\{\Sigma_\bx\}_{21}g_{1}}
\right)\right]\\
&=\frac{1}{2}\left\{\E_{\bg}\!\left[\!f\!\left(\!
\frac{\beta}{\sigma\!+\!|g_1|^2P_1\!+\!|g_2|^2P_2\!+\!2{\rm
Re}\{g_{1}^*\{\Sigma_\bx\}_{12}g_{2}\}}
\right)\right]\right.\\
&~~~~~~\left.\!+\!\E_{\bg}\!\left[\!f\!\left(\!
\frac{\beta}{\sigma\!+\!|g_1|^2P_1\!+\!|g_2|^2P_2\!-\!2{\rm
Re}\{g_{1}^*\{\Sigma_\bx\}_{12}g_{2}\}}
\!\right)\!\right]\!\right\}\\
&=\frac{1}{2}\E_{\bg}\!\!\left[\!f\!\left(\!
\frac{\beta\left[\beta+2\sigma\!+\!2\left(\!|g_1|^2P_1\!+\!|g_2|^2P_2\right)\right]}{\left(\sigma\!+\!|g_1|^2P_1\!+\!|g_2|^2P_2\right)^2\!-\!\left(2{\rm
Re}\{g_{1}^*\{\Sigma_\bx\}_{12}g_{2}\}\right)^2}
\!\right)\!\right]\!.
\end{align*}
In this case, the expectation is maximized by choosing
\[
\{\Sigma_\bx\}_{12}=\{\Sigma_\bx\}_{21}^*=0
\]
since $\beta<0$ and $\beta+2\sigma> 0$. Hence, subject to
\eqref{eq_per_equality_cons}, \eqref{eq_per_MISOSE_Cs_optimizae}
is maximized by choosing $\Sigma_\bx=\diag(P_1,P_2)$, for the case
with $n_t=2$.

(ii) Suppose that the statement holds for $n_t=k$, where $k\geq
2$. We need to show that it also holds for $n_t=k+1$.
Specifically, for $n_t=k+1$, with eavesdropper channel being
$\bg=[g_1,\ldots, g_{k+1}]^T$, the objective function
$\E_{\mathbf{g}}\!\left[\log\left(1+\frac{\beta}{\sigma+\mathbf{g}^\dag\Sigma_{\mathbf{x}}\mathbf{g}}\right)\right]$
in \eqref{eq_per_MISOSE_Cs_optimizae} can be written as
\begin{align}
\!\!& \E_{\bg}\!\!\left[\log\!\left(1+
\frac{\beta}{\sigma\!+\!\!\sum\limits_{i=1}^{k+1}\sum\limits_{j=1}^{k+1}\!
g_i^*\{\Sigma_{\mathbf{x}}\}_{ij}g_j}\!
\right)\!\right]\notag\\
&=\E_{\bg}\left[\log\left( 1+\frac{\beta}{\sigma_{k+1}\!+\!2{\rm
Re}\{\sum\limits_{i=1}^k
g_i^*\{\Sigma_{\mathbf{x}}\}_{i,k+1}g_{k+1}\}} \right)\right],
\label{eq_in_lemma2}
\end{align}
where
\[
\sigma_{k+1}\triangleq
\sigma+\tilde\bg^\dagger\tilde{\Sigma}_{\mathbf{x}}\tilde\bg+|g_{k+1}|^2P_{k+1},
\]
with $\tilde\bg=[g_1,\ldots, g_k]^T$ and
$\tilde{\Sigma}_{\mathbf{x}}$ being a $k\times k$ matrix with
$\{\tilde{\Sigma}_{\mathbf{x}}\}_{i,j}=\{\Sigma_{\mathbf{x}}\}_{i,j}$,
for $i,j=1,\ldots, k$. Then, by the fact that $-g_{k+1}$ has the
same distribution as $g_{k+1}$, the expectation in
\eqref{eq_in_lemma2} becomes
\begin{align*} \nonumber
\frac{1}{2}\E_{\bg}\left[\log \left( 1+
\frac{\beta\left(\beta\!+\!2\sigma+2(\tilde\bg^\dagger\tilde{\Sigma}_{\mathbf{x}}\tilde\bg+|g_{k+1}|^2P_{k+1})\right)}{(\sigma_{k+1})^2-(2{\rm
Re}\{\sum_{i=1}^k \!g_i^*\{\Sigma_\bx\}_{i,k+1}g_{k+1}\})^2}
\right)\right].
\end{align*}
Here, the expectation is maximized by choosing
\[
\{\Sigma_\bx\}_{i,k+1}=\{\Sigma_\bx\}_{k+1,i}^*=0
\]
for $i=1,\ldots, k$, due to the assumptions $\beta<0$ and
$\beta+2\sigma>0$. In this case, \eqref{eq_in_lemma2} becomes
\begin{align}
\E_{\bg}\left[\log \left(1+
\frac{\beta}{\sigma+|g_{k+1}|^2P_{k+1}+\tilde\bg^\dagger\tilde{\Sigma}_\bx\tilde\bg}
\right)\right]. \label{eq_in_lemma2b}
\end{align}
Then because $\sigma+|g_{k+1}|^2P_{k+1}\geq 0$,
$\beta+2(\sigma+|g_{k+1}|^2P_{k+1})
> 0$, and $\tilde{\Sigma}_\bx$ is positive semi-definite, it
follows from the inductive hypothesis that \eqref{eq_in_lemma2b}
is maximized by choosing $\tilde{\Sigma}_\bx=\diag(P_1,\ldots,
P_k)$. This concludes that subject to the equality constraint
\eqref{eq_per_equality_cons}, the optimal solution maximizing the
objective function of \eqref{eq_per_MISOSE_Cs_optimizae} is
\eqref{eq_per_optimal}.

\bibliographystyle{IEEEtran}
\bibliography{IEEEabrv,Book,SecrecyPs3}
\end{document}